\documentclass[journal,onecolumn]{IEEEtran}

\usepackage{amsmath,amssymb,amsfonts,amsthm}
\usepackage{inputenc}
\usepackage{enumerate}
\usepackage{graphicx}
\usepackage{multicol}
\usepackage{multirow}
\newtheorem{thm}{Theorem}
\newtheorem{prop}{Proposition}
\newtheorem{lem}{Lemma}
\newtheorem{cor}{Corollary}

\newtheorem{exam}{Example}

\newtheorem{rem}{Remark}

\def\0{{\mathbf 0}}

\newcommand{\F}{\mathbb{F}}
\newcommand{\Z}{\mathbb{Z}}

\begin{document}

\sloppy

\title{Explicit constructions of optimal linear codes with Hermitian hulls and their application to quantum codes}
\author
{
{
Lin Sok\thanks{This research work is supported by Anhui Provincial Natural Science Foundation with grant number 1908085MA04,
Lin Sok is with School of Mathematical Sciences, Anhui University,  230601 Anhui,  P. R. China
$\&$ Department of Mathematics, Royal University of Phnom Penh, 12156 Phnom Penh
(email: soklin\_heng@yahoo.com) }
}
}

\date{}



\maketitle


\begin{abstract}
We prove that any Hermitian self-orthogonal $[n,k,d]_{q^2}$ code gives rise to an $[n,k,d]_{q^2}$ code with $\ell$ dimensional Hermitian hull for $0\le \ell \le k$. We present a new method to construct Hermitian self-orthogonal $[n,k]_{q^2}$ codes with large dimensions $k>\frac{n+q-1}{q+1}$. New families of Hermitian self-orthogonal codes with good parameters are obtained; more precisely those containing almost MDS codes. By applying a puncturing technique to Hermitian self-orthogonal codes, MDS $[n,k]_{q^2}$ linear codes with Hermitian hull having large dimensions $k>\frac{n+q-1}{q+1}$ are also derived. 
New families of MDS, almost MDS and optimal codes with arbitrary Hermitian hull dimensions are explicitly constructed from algebraic curves. As an application, we provide entanglement-assisted quantum error correcting codes with new parameters.

\end{abstract}

\section{Introduction}
Quantum error correcting codes have attracted a lot of attention recently due to their capability to protect information carrying quantum states against decoherence in quantum information systems. A lot of effort has been done to construct good quantum codes from the classical error correcting codes.

Entanglement is one of the best approaches to achieve higher rates in quantum information systems. Entanglement-assisted quantum error-correcting codes (EAQECCs) were firstly introduced by Bowen \cite{Bow}. Later, they were developed by Brun \emph{et al.} \cite{BDH06}. The authors \cite{BDH06} showed that if pre-shared entanglement between the encoder and decoder is available, the EAQECCs can be constructed via classical linear codes without self-orthogonality as in \cite{AK01}. Then Wilde {\em et al. }\cite{WilBru} proposed two methods to construct EAQECCs from classical codes; the Euclidean and Hermitian construction methods. 
The parameters of EAQECCs can be determined from those of hulls of linear codes (see \cite{GJG18}). For the works on EAQECCs from Euclidean hulls, the reader is referred to \cite{CarMesTanQiPel18,CarLiMes,LiZeng,QCM,MesTanQi,Sok1D,Sok1D2,GGJT18,LCC}. In \cite{GGJT18}, Guenda \emph{et al.} investigated the $\ell$-intersection pair of linear codes, where they completely determined the $q$-ary MDS EAQECCs of length $n \leq q+1$ for all possible parameters. EAQECCs with zero pre-shared entanglements, known as quantum stabilizer codes, can be constructed from Hermitian linear complementary dual codes using method \cite{CarMesTanQiPel18}. Now the remaining problem is to construct new $q$-ary MDS EAQECCs of length $n > q+1$ with hull dimension $\ell>1$. This can be achieved by considering Hermitian hulls of linear codes, see for example \cite{FangFuLiZhu,PerPel}. It is has been known in \cite{FangFuLiZhu} that given $n,k$ and $q$, it is a very difficult problem to construct new MDS Hermitian hull codes with parameters $[n,k]_{q^2}$ for new lengths $n$ or  for new dimensions $k>\frac{n+q-1}{q+1}$. 

Algebraic geometry codes were invented by Goppa, where in some literature they were also called geometric Goppa codes. In his paper \cite{Goppa}, Goppa showed how to construct linear codes from algebraic curves over a finite field. A very nice property of these codes is that their parameters can be determined from the degree of a divisor associated. Despite a strongly theoretical construction, algebraic geometry (AG) codes have asymptotically good parameters, and this was the first time that linear codes were proved to have improved the so-called Gilbert-Vasharmov bound. Algebraic geometry codes are good candidates for constructing quantum codes, see for example \cite{JinXing12,Jin,PerPel,Sok1D2,SokQSC}. In \cite{PerPel}, some good EAQECCs were constructed from AG codes of arbitrary Hermitian hull dimensions, but their constructions are not explicit.

In this paper, we study linear codes with Hermitian hull using tools from algebraic function fields. We introduce a new method to construct Hermitian hull $[n,k]$ codes. We prove that any Hermitian self-orthogonal $[n,k]$ code gives rise to a linear code with Hermitian hull of dimension less than $k$ and with the same parameters as the former codes, and thus new families of MDS Hermitian hull codes are obtained in Theorem \ref{thm:large-family} 1). New constructions of Hermitian self-orthogonal codes is presented in Theorems \ref{thm:embedding-new}, \ref{thm:embedding-new2} and \ref{thm:embedding-new3}, which give a generalization of the recent result \cite{SokQSC}. From an MDS Hermitian self-orthogonal code, this new method enables us to construct (almost MDS) Hermitian self-orthogonal codes and thus Hermitian hull codes with dimension $k>\frac{n+q-1}{q+1}$. More specifically, many new almost MDS Hermitian hull codes are constructed (see Theorems \ref{thm:family-AMDS} and \ref{thm:family-AMDS2}). 
By applying a puncturing technique to Hermitian self-orthogonal codes, MDS Hermitian hull codes with large dimension $k>\frac{n+q-1}{q+1}$ are derived (see Theorem \ref{thm:family-new-mds} and \ref{thm:family-AMDS2}), and these codes are new compared with those explicitly constructed in \cite{FangFuLiZhu}. Theorem \ref{thm:hermitian-3} enlarges families of Hermitian self-orthogonal codes \cite{SokQSC}. Except for few values in Theorem \ref{thm:Q-MDS}, all families of EAQECCs obtained in Theorems \ref{thm:Q-MDS}-\ref{thm:Q-punctured} are new.

The paper is organized as follows: Section \ref{section:pre} gives preliminaries and background on algebraic geometry codes. Section \ref{section:constructions} provides a new construction method for Hermitian self-orthogonal codes and for  linear codes with arbitrary Hermitian hull dimensions. In Section \ref{section:algebraic curves}, Hermitian hull codes are constructed from projective lines, elliptic curves, hyper-elliptic curves and Hermitian curves. In the last section, we give an application of Hermitian hull codes to EAQECCs.

\section{Preliminaries}\label{section:pre}
\subsection{Linear codes}
Let $\F_q$ be the finite field with $q$ elements. A 
linear code of length $n$, dimension $k$  and minimum distance $d$ over ${{\mathbb F}_q}$ is denoted as $[n,k,d]_q$ code. If $C$ is an $[n,k,d]_q$ code, then from the Singleton bound, its minimum distance is upper bounded by
$$d\le n-k+1.$$
A code is called {\em Maximum Distance Separable} ({MDS}) if its minimum distance $d=n-k+1$ and is called {\em almost} MDS (AMDS) if it is minimum distance $d=n-k$. 
A code is called {\it optimal} if it has the highest possible minimum distance for its length and dimension.

The {\em Hermitian} inner product of ${\bf{x}}=(x_1,
\dots, x_n)$ and ${\bf{y}}=(y_1, \dots, y_n)$ in $\F_{q^2}^n$ is defined by
$$<{\bf{x}},{\bf{y}}>_h=\sum_{i=1}^n x_i y_i^q.
$$ The Hermitian {\em dual} of $C$,
denoted by $C^{\perp_h}$, is the set of vectors orthogonal to every
codeword of $C$ under the Hermitian inner product. A linear code $C$ is called Hermitian self-orthogonal if $C\subseteq C^{\perp_h}$. The Hermitian hull of a linear code $
C$ is $$Hull_h(C):=C\cap C^{\perp_h}.$$ A linear code
$C$ is called {\em Hermitian} $\ell$-$dim$ hull if $\dim (Hull_h(C))=\ell$. 

For a linear code $C\subseteq \F_q^n$ and ${\bf v}=(v_1,\hdots,v_n)\in (\F_q^*)^n$, we define

$$
{\bf v}\cdot C:=\{(v_1c_1,\hdots,v_nc_n)| {\bf c}=(c_1,\hdots,c_n)\in C\}.
$$
It is easy to see that ${\bf v}\cdot C$ is a linear code if and only if $C$ is a linear code. Moreover, both codes have the same dimension, minimum Hamming distance, and weight distribution.

For undefined terms related to algebraic function fields, the reader is referred to Stichtenoth \cite{Stich}. 

Let ${\cal X}$ be a smooth projective curve of genus $g$ over $\F_q.$
We denote the field of rational functions of ${\cal X}$ by $\F_q({\cal X}).$ Function fields of
algebraic curves over a finite field are
finite separable extensions of $\F_q(x)$. Points on the curve ${\cal X}$ identified with places of the
function field $\F_q({\cal X}).$ We call a point on ${\cal X}$ rational if all of
its coordinates belong to $\F_q.$ Rational points can be identified
with places of degree one. The set of $\F_q$-rational
points on ${\cal X}$ is denoted by  ${\cal X}(\F_q)$.

We define a divisor $G$ on the curve ${\cal X}$ to be the formal sum $\sum\limits_{P\in {\cal X}}n_PP$ with only finitely many nonzeros $n_P\in \Z$. A divisor $G$ on $\cal X$ is called rational if for any $\sigma$ in $\text{Gal}(\overline{\F}_q/\F_q)$, we have $G^\sigma=G.$ 
We define the support of $G$ as $supp(G):=\{P\in {\cal X}|n_P\not=0\}$. For $G=\sum\limits_{P\in {\cal X}}n_PP$, the degree of $G$ is defined by $\deg(G):=\sum\limits_{P\in {\cal X}}n_P\deg(P)$, where $\deg (P)=|P^\sigma |$ is the size of orbit of $P$ under the action $\sigma$. 
For two divisors $G=\sum\limits_{P\in {\cal X}}n_PP$ and $H=\sum\limits_{P\in {\cal X}}m_PP$,  we say that $G\ge H$ if $n_P\ge m_P$ for all places $P\in {\cal X}$. 

For a nonzero rational function $f$ on the curve $\cal X$, we define the ``principal" divisor of $f$ as
$(f):=\sum\limits_{P\in {\cal X}}v_P(f)P,$ where $v_P$ denotes the normalized discrete valuation corresponding to the place $P$.
If $Z(f)$ and  $N(f)$ denotes the set of zeros and poles of $f$ respectively, we define the zero divisor and pole divisor of $f$, respectively by
$
\begin{array}{c}
(f)_0:=\sum\limits_{P\in Z(f)}v_{P}(f)P,
(f)_\infty:=\sum\limits_{P\in N(f)}-v_{P}(f)P.\\
\end{array}
$

For a divisor $G$ on the curve $\cal X$, we define
$${\cal L}(G):=\{f\in \F_q({\cal X})\backslash \{0\}|(f)+G\ge 0\}\cup \{0\},$$ 
and 
$${\Omega}(G):=\{\omega\in \Omega\backslash \{0\}|(\omega)-G\ge 0\}\cup \{0\},$$
where $\Omega:=\{fdx|f\in \F_q({\cal X})\}$, the set of differential forms on $\cal X$. It is well-known that both ${\cal L}(G)$ and ${\Omega}(G)$ are finite dimensional vector spaces. Moreover, for a differential form $\omega$ on $\cal X$, there exists a unique a rational function $f$ on $\cal X$ such that $\omega=fdt,$
where $t$ is a local uniformizing parameters. In this case, we define the divisor associated to $\omega$ by $(\omega)=\sum\limits_{P\in {\cal X}}v_P(\omega)P, $
where $v_P(\omega):=v_P(f).$ The divisor class of a nonzero differential form is called the canonical divisor. It is well-known that if $K$ is the canonical divisor, then $\deg (K)=2g$.

The dimension of ${\cal L}(G)$ is determined by Riemann-Roch's theorem as follows.

\begin{prop}\textnormal{\cite[Theorem 1.5.15 (Riemann-Roch)]{Stich}}\label{thm:Riemann-Roch} Let $W$ be a canonical divisor. Then, for each divisor $G$, the following holds:
$$\text{dim } {\cal L}(G) = \deg G + 1 -g + \text{dim } {\cal L}(W-G),$$
where $g$ is the genus of the smooth algebraic curve.
\end{prop}

Let $D=P_{\alpha_1}+\cdots+P_{\alpha_n}$, where $(P_{\alpha_i})_{1\le i \le n}$ are places of degree one, and $G$ a divisor with $supp(D)\cap supp(G)=\emptyset$. Define the algebraic geometry code by
$$
C_{\cal L}(D,G):=\{(f(P_{\alpha_1}),\hdots,f(P_{\alpha_n}))|f\in {\cal L}(G)\},
$$
and the differential algebraic geometry code as
$$
C_{\Omega}(D,G):=\{(\text{Res}_{P_{\alpha_1}}(\omega),\hdots,\text{Res}_{P_{\alpha_n}}(\omega))|\omega\in {\Omega}(G-D)\},
$$
where $\text{Res}_{P}(\omega)$ denotes the residue of $\omega$ at point $P.$

For ${\bf v}=(v_1,\hdots,v_n)$ with $v_i\in \F_{q^2}^*,$ define also
\begin{equation}
{\cal L}_q(G):=\{f\in {\cal L}(G)|f(\alpha)\in \F_q,\forall \alpha\in \F_{q^2}\},
\end{equation}
\begin{equation}
C_{{\cal L}_{q}}(D,G):=\{(f(P_{\alpha_1}),\hdots,f(P_{\alpha_n})|f\in {{\cal L}_{q}}(G)\},
\end{equation} 
and
\begin{equation}
C_{{\cal L}_{q}}(D,G;{\bf v}):=\{(v_1f(P_{\alpha_1}),\hdots,v_nf(P_{\alpha_n}))|f\in {\cal L}_{q}(G)\}.
\end{equation}

The parameters of an algebraic geometry code $C_{\cal L}(D,G)$ is given as follows.
\begin{prop}\textnormal{\cite[Corollary 2.2.3]{Stich}}\label{thm:distance} Assume that $2g -2 < deg(G) < n.$ Then the code $C_{\cal L}(D,G)$  has parameters $[n,k,d]$ satisfying
\begin{equation}
k=\deg (G)-g+1\text{ and } d\ge n-\deg (G).
\label{eq:distance}
\end{equation}

\end{prop}

\subsection{Quantum codes}

Let  $({\mathbb C})^{\bigotimes n}$ ($\cong \mathbb{C}^{q^n}$) be the $q^n$-dimensional Hilbert space over the complex field $\mathbb{C}$. A quantum code of length $n$ is a subspace of $\mathbb{C}^{q^{n}}$.
Let $\{|\textbf{a}\rangle =|a_{1}\rangle\bigotimes|a_{2}\rangle\bigotimes\cdots\bigotimes|a_{n}\rangle: (a_{1}, a_{2}, \ldots, a_{n}) \in \mathbb{F}^{n}_{q}\}$ be a basis of $\mathbb{C}^{q^{n}}$. We define the inner product of two quantum states
$
|\phi_1\rangle=\sum\limits_{{\bf a}\in \F_{q}^{n}}\phi_1({\bf a})|{\bf a}\rangle 
\text { and } |\phi_2\rangle=\sum\limits_{{\bf a}\in \F_{q}^{n}}\phi_2({\bf a})|{\bf a}\rangle
$
by
\begin{equation*}
\langle \phi_1|\phi_2\rangle=\sum\limits_{{\bf a}\in \F_{q}^{n}}
\overline{\phi_1({\bf a})}\phi_2({\bf a})\in {\mathbb C} \text{,} 
\end{equation*}
where  $\overline{\phi_1({\bf a})}$ is the complex conjugate of 
$\phi_1({\bf a})$. Two quantum states $|\phi_1\rangle$ and $|\phi_2\rangle$ are 
\textit{orthogonal} if $\langle \phi_1|\phi_2\rangle=0$.

The rules of $X(\textbf{a})$ and $Z(\textbf{b})$ on $|\textbf{v}\rangle \in \mathbb{C}^{q^{n}}$ ($\textbf{v} \in \mathbb{F}^{n}_{q}$) are given as
\begin{equation}
X(\textbf{a})|\textbf{v}\rangle=|\textbf{v}+\textbf{a}\rangle
\textnormal{ and }
Z(\textbf{b})|\textbf{v}\rangle=\zeta_{p}^{tr(\langle \textbf{v}, \textbf{b} \rangle_{E})}|\textbf{v}\rangle,
\label{eq:action}
\end{equation}
respectively, where $tr(\cdot)$ is the trace function from $\mathbb{F}_{q}$ to $\mathbb{F}_{p}$ and $\zeta_{p}$ is a complex primitive $p$-th root of unity. 

For ${\bf a}=(a_1,\ldots,a_{n}),{\bf b}=(b_1,\ldots,b_{n}) \in \F_{q}^{n}$, we can also write, from (\ref{eq:action}), $X({\bf a}) = X(a_{1}) \otimes \ldots \otimes X(a_{n})$ and
$Z({\bf b}) = Z(b_{1}) \otimes \ldots \otimes Z(b_{n})$ for the tensor product of
$n$ (error) operators. The set ${\cal E}_{n}\{ X({\bf a})Z({\bf b}) : {\bf a},{\bf b} \in \F_{q}^{n} \}$
is an error basis on $\mathbb{C}^{q^{n}}$.
The error group $G_{n}$ is generated by the matrices in ${\cal E}_{n}$ and is defined by
\begin{equation*}
 G_{n}:=\{\zeta_{p}^{t} X({\bf a}) Z({\bf b}) : {\bf a},{\bf b} \in \F_{q}^{n}, t \in \F_{p} \} \text{.}
\end{equation*}
For $E = \zeta_{p}^{t} X({\bf a}) Z({\bf b}) \in G_{n}$, we define the \textit{quantum weight}
${\bf wt}_{Q}(E)$ of $E$ to be the number of coordinates such that 
$(a_{i},b_{i}) \neq (0,0)$. 
A quantum code $Q$ is called a $d-1$ quantum error detecting code ($d\geq1$) if, for any pair $|\phi_1\rangle$ and $|\phi_2\rangle$ in
$Q$ with $\langle \phi_1 |\phi_2 \rangle=0$ and any $E \in G_{n}$ with ${\bf wt}_{Q}(E) \leq d-1$, 
$|\phi_1\rangle$ and $E|\phi_2\rangle$ are orthogonal. The quantum code $Q$ is said to have minimum distance $d$ if $d$ is the largest integer such that for any $|\phi_1\rangle, |\phi_2\rangle\in Q $ with $\langle \phi_1 |\phi_2 \rangle=0$ and any $E\in G_n$ with ${\bf wt}_{Q}(E) \le d-1$, we have that $ |\phi_1 \rangle$ and $E|\phi_2\rangle$ are orthogonal.
We write $[[n,k,d]]_{q}$ for a $q$-ary quantum code of length $n$, dimension $k$ and minimum distance $d$. It is well-known that a quantum code with minimum distance $d$ can detects up to $d-1$ errors and correct  up to $\frac{d-1}{2}$ quantum errors.

For $S$ being an abelian subgroup of $G_n$, we define the quantum stabilizer codes $C(S)$ by
$$
C(S):=\{|\phi \rangle : E|\phi \rangle=|\phi \rangle,\forall E\in S\}.
$$
It was shown in \cite{Cal Rai,AK01,Ket Kla} that such codes can be constructed from classical linear codes with some properties of self-orthogonality.
However, the above construction method is not applicable anymore if the subgroup $S$ of $G_n$ is non-abelian. To improve the construction, Brun \emph{et al.} \cite{BDH06}  introduced the so-called entanglement-assisted quantum error-correcting codes (EAQECCs). In their method, they extended $S$ to be a new abelian subgroup in a larger error group, and they assumed that a sender and a receiver shared a certain amount of pre-existing entangled bits (ebits), which was not subject to errors. 

We use $[[n, k, d; c]]_{q}$ to denote a $q$-ary $[[n, k, d]]_{q}$ quantum code that utilizes $c$ pre-shared entanglement pairs. 
For $c=0$, an $[[n, k, d; c]]_{q}$ EAQECC is equivalent to a quantum stabilizer code \cite{AK01}.

 The constraints among the parameters of an EAQECC are given in the following lemma \cite{LA18}.
\begin{prop}\textnormal{(Quantum Singleton Bound)}\label{lem4.1}
  For any $[[n, k, d; c]]_{q}$-EAQECC, if $d \leq \frac{n+2}{2}$, we have
  \[ 2(d-1)\le n+c-k .\]
\end{prop}
When the bound meets with equality, the EAQECC is called MDS, and it is called AMDS if its minimum distance is one unit less than the MDS case.

We denote the conjugate transpose of a matrix $M$ by $M^{\dag}$, that is, if $M=m_{i,j}$ then $M^{\dag}=m_{j,i}^q.$ 
\begin{lem} \textnormal{(\cite{BDH06})}\label{lem:Q-construction}
Let $P$ be the parity check matrix of an $[n, k, d]_{q^2}$ code $C$.  Then there exists an $[[n, 2k - n + c, d; c]]_{q}$ EAQECC $\cal Q$, where $c =\text{rank}(PP^{\dag})$ is the required number of maximally entangled states. In particular, if $C$ is an MDS code and $d \leq \frac{n+2}{2}$, then $\mathcal{Q}$ is an MDS EAQECC.
\end{lem}

\begin{lem}\textnormal{\cite{GJG18}}
\label{lem:hull-H}
Let $C$ be a classical $[n,k,d]_{q^2}$ code with parity check matrix $P$ and generator matrix $G$.
Then $\textnormal{rank}(PP^{\dag})$ and $\textnormal{rank}(GG^{\dag})$ are independent of $P$ and $G$ so that
$$
\begin{array}{ll}
\textnormal{rank}(PP^{\dag})&=n-k-\dim(Hull_h(C)) \\
&= n-k-\dim(Hull_h(C^{\perp_h})),
\end{array}
$$
and
$$
\begin{array}{ll}
\textnormal{rank}(GG^{\dag})&=k-\dim(Hull_h(C)) \\
&= k-\dim(Hull_h(C^{\perp_h})).
\end{array}
$$
\end{lem}

\section{New constructions of Hermitian $\ell$-$dim$ hull codes}\label{section:constructions}

In this section, we first show that any Hermitian $\ell$-$dim$ hull $[n,k,d]_{q^2}$ code can be constructed from a Hermitian self-orthogonal $[n,k,d]_{q^2}$ code. Then we present a new method to embed a Hermitian self-orthogonal code into another one.

\begin{lem} \label{lem:char2} Assume that there exists a Hermitian self-orthogonal code with parameters $[n,k,d]_{q^2}$. Then there exists a Hermitian $\ell$-$dim$ hull $[n,k,d]_{q^2}$ code for $0\le \ell \le k$.
\end{lem}
\begin{proof} Assume that $C$ is a Hermitian self-orthogonal code with parameters $[n,k,d]$ over $\F_{q^2}$. Let $G$ be the generator matrix of $C$. Then up to equivalence, we can write $G$ as 

$$G=(I_k|A),$$
where $A$ satisfies $AA^{\dag}=-I_k.$ Choose $\lambda\in \F_{q^2}$ such that $\lambda^{q+1} \not =1$, and take 
$G_\ell=\text{diag}(\underbrace{ \lambda ,\hdots,\lambda }\limits_{k-\ell}, \underbrace{ 1,\hdots,1}\limits_{\ell}|A)$. Then it is easy to see that $\textnormal{rank}(G_\ell G_\ell^{\dag})$ is exactly $k-\ell$, and thus the result follows from Lemma \ref{lem:hull-H}.
\end{proof}

We now state sufficient conditions for an algebraic geometry $C_{{\cal L}_{q}}(D,G;{\bf v})$ to be Hermitian self-orthogonal and introduce new constructions of such codes.
\begin{lem}\textnormal{\cite{SokQSC}}\label{thm:key} Let $D=P_{\alpha_1}+\cdots+P_{\alpha_n}$ and $G=(k-1)O$ be two divisors, $\omega$ be a Weil differential form such that $H=D-G+(\omega)$ and ${\bf v}=(v_1,\hdots,v_n)$ with $v_i\in \F_{q^2}^*.$ Then the code $C_{{\cal L}_{q}}(D,G;{\bf v})$ is Hermitian self-orthogonal if the following conditions hold
\begin{enumerate}
\item $G\le H$,
\item $\text{Res}_{P_{\alpha_i}}(\omega)=v_i^{q+1}$,
\item $g+1\le k\le \lfloor \frac{n+q+2g-1}{q+1} \rfloor$, where $g$ is the genus of the defining curve.
\end{enumerate}
\end{lem}

The following lemma is useful for embedding an Hermitian self-orthogonal code into another one.
\begin{lem}\textnormal{\cite{SokQSC}}\label{lem:embedding} Assume that $D=P_{\alpha_1}+\cdots+P_{\alpha_n}$, $G=(k-1)O$, $1\le k\le \lfloor \frac{n+q-1}{q+1} \rfloor$ and $\text{Res}_{P_{\alpha_i}}(\omega)=v_i^{q+1}$ for some $v_i \in  \F_{q^2}^*$ for $1\le i\le n$.  Then an MDS Hermitian $q^2$-ary self-orthogonal code $C_{{\cal L}_q}(D,G;{\bf v})$ with parameters $[n,k]$ can be embedded into an MDS Hermitian $q^2$-ary self-orthogonal $[n+1,k+1]$ code.
\end{lem}
In \cite{SokQSC}, MDS Hermitian self-orthogonal codes are constructed up to dimension $q$. The following theorem extends the embedding Lemma \ref{lem:embedding}.
\begin{thm}\label{thm:embedding-new}Assume that $(n-1)|(q^2-1)$, $D=P_{\alpha_1}+\cdots+P_{\alpha_n}$, $G=(k-1)O$, $1\le k\le \lfloor \frac{n+q-1}{q+1} \rfloor$ and $\text{Res}_{P_{\alpha_i}}(\omega)=v_i^{q+1}$ for some $v_i \in  \F_{q^2}^*$ for $1\le i\le n$. Put $n_i=n+i$ and $k_i=k+i$ for $1\le i\le q$. Then an MDS Hermitian $q^2$-ary self-orthogonal code $C_{{\cal L}_q}(D,G;{\bf v})$ with parameters $[n,k]$ can be embedded into a Hermitian $q^2$-ary self-orthogonal code $C_i$ with parameters 
\begin{enumerate}
\item $[n+1,k+2,\ge n-k-1]_{q^2}$ if $(n-1)|(k+1)(q+1)$;
\item $[n_i-1,k_i,\ge n-k-i+1]_{q^2}$ for $2\le i\le q$ if $(n-1)|k(q+1)$.
\end{enumerate}
\end{thm}
\begin{proof}Take $U_{n-1}=\{ \alpha\in \F_{q^2}|\alpha^{n-1}=1\}=\{\alpha_1,\hdots,\alpha_{n-1}\}$, and set $U=\{\alpha_1,\hdots,\alpha_{n-1},\alpha_n\}$, where $\alpha_n=0$. It should be noted that under the assumption in the theorem, there exists an MDS Hermitian self-orthogonal code $C_{{\cal L}_q}(D,G;{\bf v})$ with parameters $[n,k]_{q^2}$, where $D=\sum\limits_{\alpha\in U}P_\alpha=P_{\alpha_1}+\cdots+P_{\alpha_n}$. 
Consider the code $C_i$ with its generator matrix ${\cal G}_i$ written as follows:
{\scriptsize
\begin{equation}
{\cal G}_i=\left(
\begin{array}{llllllll}
v_1&\hdots&v_{n}&0&0&\cdots&0\\
v_1\alpha_1&\cdots&v_{n}\alpha_{n}&0&0&\cdots&0\\
\vdots&\cdots&\vdots&0&0&\cdots&0\\
v_1\alpha_1^{k-1}&\cdots&v_{n}\alpha_{n}^{k-1}&0&0&\cdots&0\\
v_1\alpha_1^{k}&\cdots&v_{n}\alpha_{n}^{k}&\lambda_1&0&\cdots&0\\
v_1\alpha_1^{k+1}&\cdots&v_{n}\alpha_{n}^{k+1}&0&\lambda_2&\cdots&0\\
\vdots&\vdots&\vdots&\vdots&\vdots&\cdots&\vdots\\
v_1\alpha_1^{k+i-1}&\cdots&v_{n}\alpha_{n}^{k+i-1}&0&0&\cdots&\lambda_i\\
\end{array}
\right).
\label{eq:Gi}
\end{equation}
}
Note that the first $k$ rows of ${\cal G}_i$ generates $C_{{\cal L}_q}(D,G;{\bf v})$. We now calculate the minimum distance of $C_i$.
First, observe that by puncturing the last $i$ coordinates of $C_i$, we obtain an MDS code with parameters $[n,k_i,n-k_i+1]$. Hence, the code $C_i$ has parameters $[n_i,k_i,\ge n-k-i+1]$. 
\begin{enumerate}
\item for $i=2$, consider the following system of $(k+1)(q+1)+1$ equations with $n$ indeterminates $v_1^{q+1},\hdots,v_n^{q+1}$ defined by

\begin{equation}
M{\bf v}={\bf a},
\label{eq:lambdai}
\end{equation}
where 
{
$$
M=\left(
\begin{array}{cccc}
1&\cdots&1\\
\alpha_1&\cdots&\alpha_n\\
\vdots&\cdots&\vdots\\
\alpha_1^{(k-1)(q+1)}&\cdots&\alpha_n^{(k-1)(q+1)}\\
\vdots&\cdots&\vdots\\
\alpha_1^{k(q+1)}&\cdots&\alpha_n^{k(q+1)}\\
\vdots&\cdots&\vdots\\
\alpha_1^{(k+1)(q+1)}&\cdots&\alpha_n^{(k+1)(q+1)}\\
\end{array}
\right),
{\bf v}=\left(
\begin{array}{c}
v_1^{q+1}\\
v_2^{q+1}\\
\vdots\\
v_n^{q+1}
\end{array}
\right),
{\bf a}=
\left(
\begin{array}{c}
0\\
0\\
\vdots\\
0\\
-\lambda_1^{q+1}\\
0\\
\vdots\\
0\\
-\lambda_2^{q+1}\\
\end{array}
\right).
$$
}
Then $C_2$ is Hermitian self-orthogonal if and only if (\ref{eq:lambdai}) has a solution. It should be noted that $\lambda_1$ does exist by Lemma  \ref{lem:embedding}. 
Since the first $k+1$ rows $g_1,\hdots,g_k,g_{k+1}$ of ${\cal G}_2$ generates a Hermitian self-orthogonal $[n+2,k+1]$ code, it remains to check that $<g_{k+1},g_{k+2}>_h=0$ and $<g_{k+2},g_{k+2}>_h=-\lambda_2^{q+1}$ for some $\lambda_2\in \F_{q^2}.$
Now under the condition $(n-1)|(k+1)(q+1)$, we have that $(\alpha_1^{(k+1)(q+1)},\cdots, \alpha_n^{(k+1)(q+1)})=(1,\hdots,1,0)$, and so $v_1^{q+1}\alpha_1^{(k+1)(q+1)}+\cdots+ v_n^{q+1}\alpha_n^{(k+1)(q+1)}=-\lambda_2^{q+1}$ has a solution. Thus, $\lambda_2$ is well determined (here $\lambda_2=v_n$). Now, we check that $<g_{k+1},g_{k+2}>_h=0$, that is, 
\begin{equation}
\sum\limits_{i=1}^n\alpha_i^{kq+k+1}v_i^{q+1}=0.
\label{eq:gkandgk1}
\end{equation}
\begin{enumerate}
\item If $k(q+1)+1\le n$, then following the same reasoning as in the proof of \cite[Lemma 4]{SokQSC}, (\ref{eq:gkandgk1}) holds with $v_i^{q+1}=\text{Res}_{P_{\alpha_i}}(\omega)$.
\item  If $k(q+1)+1 > n$, then we can write $k(q+1)+1=(n-1)A+B$ with $0\le B<(n-1)<k(q+1)$ and $\alpha_i^{kq+k+1}=\alpha_i^{B}$, and thus (\ref{eq:gkandgk1}) also holds.
\end{enumerate}
Hence, the system (\ref{eq:lambdai}) has a solution. Moreover, with $\lambda_2=0$, one can puncture the $(n+2)$-th coordinate and obtain the code with parameters $[n+1,k+2,\ge n-k-1]$, and this completes the proof.
\item for $i>2$, with $(n-1)|k(q+1)$, the Hermitian self-orthogonality follows from the same reasoning as point 1). The dimension of $C_i$ can be taken up to $k+q$.
\end{enumerate}
\end{proof}

\begin{rem}\label{rem:1} Given a Hermitian self-orthogonal code $C_i$ with its generator matrix ${\cal G}_i$ determined by (\ref{eq:Gi}) and $(g_t)_{1\le t\le n}$ being its row, it is not difficult to see that if $\gamma_{t-1}:=<g_t,g_t>_h-\sum\limits_{j=1}^nv_j^{q+1}\alpha_j^{(t-1)(q+1)}$, then 
\begin{enumerate}
\item $\gamma_0=\cdots=\gamma_{k-1}=0$;
\item $\gamma_{t}=\gamma_{t+q-1}$ for $1\le t\le k-1$;
\item $\gamma_{k-1+i}=\gamma_{k-1+i+q-1}=\lambda_i^{q+1}=\lambda_{i+q-1}^{q+1}$ for $1\le i \le q$;
\item for $1\le i\le q$, there are at most $q-k+1$ non-zero $\lambda_i$.
\end{enumerate}
For some special vaule $k$, we can easily count the number of $i$'s such that $\lambda_i= 0$. For instance for $k=q-1$, we have $\gamma_0=\gamma_1=\cdots=\gamma_{q-2}=\gamma_{q}=\cdots=\gamma_{2q-3}=0$ and $\gamma_{q-1}=\gamma_{2q-2}=\lambda_1$.
\end{rem}
In the sequel for $(\lambda_j)_{1\le j\le i}$ defined by (\ref{eq:Gi}), we denote
\begin{equation}
I=\{j|\lambda_j\not =0\},
\end{equation}
and write $\sharp I$ for the size of $I$.

\begin{cor}\label{cor:explicit-amds}  We have the following existence:
\begin{enumerate}
\item there exist AMDS Hermitian self-orthogonal codes with parameters $[q^2+1,k]_{q^2}$ for $k=q+1,\hdots,2q-2$;
\item  there exists an AMDS Hermitian self-orthogonal code with parameters $[2(q+1)+2,k+2]_{q^2}$;
\item there exist Hermitian self-orthogonal codes with parameters $[2(q+1)+1+i-1,k_i,\ge n-k-i+1]_{q^2}$ for $2\le i\le q$.
\end{enumerate}
\end{cor}
\begin{proof}First note that
if $(n-1)|(q^2-1)$, $(n-1)|k(q+1)$, and if there exists an MDS Hermitian self-orthogonal $[n,k]_{q^2}$ code with $k\le \frac{n+q-1}{q+1}$, then the following holds:
\begin{itemize}
\item $\sum\limits_{i=1}^nv_i^{q+1}\alpha_i^{k(q+1)}=\sum\limits_{i=1}^nv_i^{q+1}\alpha_i^{k'(q+1)}\not=0\text{ if } k| k'$;\\
\item $\sum\limits_{i=1}^nv_i^{q+1}\alpha_i^{k'(q+1)}=0\text{ if } k\not| k'$.\\
\end{itemize}
For the proof of 1), we take $n-1=q^2-1$ and $k=q-1$. Then from Theorem \ref{thm:embedding-new} 2), the codes $C_i$ is Hermitian self-orthogonal. Moreover, with $n-1=q^2-1$ and $k=q-1$, we get that $\lambda_i=0$ for $i=2,\hdots,q-1$. Finally, by puncturing the zero coordinates in the generator matrix ${\cal G}_i$, we get the result as claimed.\\
For the proof of 2) and 3), we take $n-1=2(q+1)$. In this case, the largest dimension $k$ equals to $2$. Moreover, we have that $\lambda_i\not=0$ for $i$ odd, and $\lambda_i=0$ for $i$ even.
\end{proof}
We illustrate the constructions of Hermitian self-orthogonal codes in the above corollary with the following two examples and give some parameters of Hermitian self-dual codes as follows:
$[6, 3, 4]_{2^2}$,
$[8, 4, 4]_{3^2}$,
$[10, 5, 6]_{4^2}$,
$[12, 6, 5] _{5^2}$,
$[16, 8, 6] _{7^2}$,
$[18, 9, 6] _{8^2}$,
$[20, 10, 7] _{9^2}$,
$[24, 12, 7] _{11^2}$.
These parameters can be obtained by embedding an MDS Hermitian self-orthogonal $[q+2,1]$ code.
\begin{exam} \label{ex:1}Take $q=5$, $n=25$. Then from \cite{SokQSC}, there exists a Hermitian self-dual code $C$ with parameters $[25,4,22]_{25}$. Using the above embedding lemma, we obtain, from the code $C$, Hermitian self-orthogonal codes $C_j$ with parameters $[25+j,4+j,\ge 22-j]_{25}$ for $1\le j\le 5$. 
Using Magma \cite{Mag}, we give the generator matrix of the Hermitian self-orthogonal code $C_j$ for $j=5.$
{\scriptsize
$$\left(
\begin{array}{rlllllll}
\beta \beta \beta \beta \beta \beta \beta \beta \beta \beta \beta \beta \beta \beta \beta \beta
    \beta \beta \beta \beta \beta \beta \beta \beta \beta& 0 0 0 0 0\\
0 \alpha w^{11} 4 \delta w^{14} \theta w^{16} w^{17} 3 w^{19} \gamma w^{21} \beta w^{23} 1 w \lambda w^3
    w^4 w^5 2 w^7 w^8 w^9& 0 0 0 0 0\\
0 \beta 1 \lambda w^4 2 w^8 \alpha 4 w^{14} w^{16} 3 \gamma \beta 1 \lambda w^4 2 w^8 \alpha 4 w^{14}
    w^{16} 3 \gamma& 0 0 0 0 0\\
0 \alpha \delta w^{16} w^{19} \beta w w^4 w^7 \alpha \delta w^{16} w^{19} \beta w w^4 w^7 \alpha
    \delta w^{16} w^{19} \beta w w^4 w^7& 0 0 0 0 0\\
0 \beta \lambda 2 \alpha w^{14} 3 \beta \lambda 2 \alpha w^{14} 3 \beta \lambda 2 \alpha w^{14} 3 \beta \lambda
    2 \alpha w^{14} 3 &\lambda 0 0 0 0\\
0 \alpha \theta \gamma w 2 w^{11} w^{16} w^{21} \lambda w^7 4 w^{17} \beta w^3 w^8 \delta 3 w^{23} w^4
    w^9 w^{14} w^{19} 1 w^5& 0 0 0 0 0\\
0 \beta w^4 \alpha w^{16} \beta w^4 \alpha w^{16} \beta w^4 \alpha w^{16} \beta w^4 \alpha w^{16}
    \beta w^4 \alpha w^{16} \beta w^4 \alpha w^{16}& 0 0 0 0 0\\
0 \alpha w^{17} 1 w^7 w^{14} w^{21} w^4 w^{11} 3 w w^8 \theta \beta w^5 4 w^{19} \lambda w^9 w^{16}
    w^{23} 2 \delta \gamma w^3& 0 0 0 0 0\\
0 \beta 2 w^{14} \beta 2 w^{14} \beta 2 w^{14} \beta 2 w^{14} \beta 2 w^{14} \beta 2 w^{14} \beta
    2 w^{14} \beta 2 w^{14} &0 0 0 0 \lambda\\
\end{array}
\right),
$$
}
where $w$ is a primitive element of $\F_{q^2}$, $\lambda=w^2$, $\alpha=w^{10}$, $\delta=w^{13}$, $\theta=w^{15}$, $\gamma=w^{20}$ and $\beta=w^{22}$.

The code $C_5$ has parameters $[30,9,18]_{25}$. By puncturing the $27$-th, $28$-th and $29$-th coordinates, we obtain a Hermitian self-orthogonal code with parameters $[27,9,18]_{25}$. From the generator matrix of the code, we can also obtain Hermitian self-orthogonal codes with parameters 
$[26, 5, 22]_{25}$, $[27, 6, 20]_{25}$, $[28, 7, 19]_{25}$, $[29, 8, 18]_{25}$, and by puncturing the zero coordinates, we get the codes with parameters $[26, 5, 22]_{25}$, $[26, 6, 20]_{25}$, $[26, 7, 19]_{25}$, $[26,8,18]_{25}$.
\end{exam}

\begin{exam}\label{ex:2} Take $q=7$, $n=17$. Then from \cite{SokQSC}, there exists a Hermitian self-dual code $C$ with parameters $[17,2,16]_{49}$. Using the above embedding lemma, we obtain, from the code $C$, Hermitian self-orthogonal codes $C_j$ with parameters $[17+j,2+j,\ge 16-j]_{49}$ for $1\le j\le 7$. Using Magma \cite{Mag}, we give the generator matrix of the Hermitian self-orthogonal code $C_j$ for $j=7.$
{\scriptsize
$$\left(
\begin{array}{rl}
w^{45} \beta \beta \beta \beta \beta \beta \beta \beta \beta \beta \beta \beta \beta \beta \beta
    \beta &0 0 0 0 0 0 0\\
0 w^{22} w^{25} w^{28} w^{31} w^{34} w^{37} 5 w^{43} \beta w w^4 w^7 w^10 w^{13} 2 w^{19}& 0 0 0 0
    0 0 0\\
0 \beta w^4 w^{10} 2 w^{22} w^{28} w^{34} 5 \beta w^4 w^{10} 2 w^{22} w^{28} w^{34} 5& \lambda 0 0 0 0
    0 0\\
0 w^{22} w^{31} 5 w w^{10} w^{19} w^{28} w^{37} \beta w^7 2 w^{25} w^{34} w^{43} w^4 w^{13}& 0 0 0 0
    0 0 0\\
0 \beta w^{10} w^{22} w^{34} \beta w^{10} w^{22} w^{34} \beta w^{10} w^{22} w^{34} \beta w^{10} w^{22}
    w^{34}& 0 0 \lambda 0 0 0 0\\
0 w^{22} w^{37} w^4 w^{19} w^{34} w 2 w^{31} \beta w^{13} w^{28} w^{43} w^{10} w^{25} 5 w^7& 0 0 0 0
    0 0 0\\
0 \beta 2 w^{34} w^4 w^{22} 5 w^{10} w^{28} \beta 2 w^{34} w^4 w^{22} 5 w^{10} w^{28}& 0 0 0 0 \lambda
    0 0\\
0 w^{22} w^{43} 2 w^{37} w^{10} w^{31} w^4 w^{25} \beta w^{19} 5 w^{13} w^{34} w^7 w^{28} w& 0 0 0 0
    0 0 0\\
0 \beta w^{22} \beta w^{22} \beta w^{22} \beta w^{22} \beta w^{22} \beta w^{22} \beta w^{22} \beta
    w^{22}& 0 0 0 0 0 0 \lambda\\
    \end{array}
\right),
$$
}
where $w$ is a primitive element of $\F_{q^2}$, $\lambda=w^3$ and $\beta=w^{46}$.

The code $C_7$ has parameters $[24,9,10]_{49}$. By puncturing the $19$-th, $21$-th and $23$-th coordinates, we obtain a Hermitian self-orthogonal code with parameters $[21,9,10]_{49}$. From the above matrix, we can also obtain Hermitian self-orthogonal codes with parameters 
$[18,3,16]_{49}$, $[19,4,14]_{49}$, $[20,5,14]_{49}$, $[21,6,12]_{49}$, $[22,7,12]_{49}$ $[23,8,10]_{49}$, and by puncturing the zero coordinates, we get the codes with parameters $[18,3,16]_{49}$, $[18,4,14]_{49}$, $[19,5,14]_{49}$, $[19,6,12]_{49}$, $[20,7,12]_{49}$, $[20,8,10]_{49}$, respectively. 
\end{exam}

\begin{thm}\label{thm:embedding-new2} Let $N_j=j(q+1)$ and $U_j=\{\alpha\in \F_{q^2}|\alpha^{N_j}=1\}$. Assume that $a=3$ or $a=4$ and $a|(q-1)$. Let $\beta_1\in U_2\backslash U_1$ and $\beta_2\in U_a\backslash U_2$. Put $U_1=\{u_1,\hdots,u_{q+1}\}$ and $U=U_1\cup \beta_1U_1\cup \beta_2U_1\cup \{0\}$.  Put $n=3(q+1)+1$, $k_i=3+i$ and $n_i=n+i$. Then 
\begin{enumerate}
\item there exists an AMDS Hermitian self-orthogonal  $[n+2,5]_{q^2}$ code;
\item  there exists a Hermitian self-orthogonal  $[n+i,3+i,\ge n-2-i]_{q^2}$ code.
\end{enumerate}
\end{thm}

\begin{proof}Denote $U=\{\alpha_1,\hdots,\alpha_{n-1},\alpha_n\}$, where $\alpha_n=0$. First note that if $a|(q-1)$, then there exist $\beta_1\in U_2\backslash U_1$ and $\beta_2\in U_a\backslash U_2$ such that $\beta_1^{2(q+1)}=1$ and $\beta_2^{a(q+1)}=1$. Set $D=\sum\limits_{\alpha\in U}P_\alpha$ and $G=(k-1)O$. Now under the assumption above, there exists, from \cite[Construction 7]{SokQSC}, an MDS Hermitian self-orthogonal code $C_{{\cal L}_q}(D,G;{\bf v})$ with parameters $[n,k]_{q^2}$, where $v_i^{q+1}=\text{Res}_{P_{\alpha_i}}(\omega)$ for $1\le i \le n$. In this case, the maximum value of the code dimension is $k=3.$  We now embed this code into another Hermitian self-orthogonal code. Consider the code $C_i$ with its generator matrix ${\cal G}_i$ written given by (\ref{eq:Gi}).
With the same reasoning as in the proof of Theorem, we get that the code $C_i$ has parameters $[n_i,k_i,\ge n-k-i+1]$. 
\begin{enumerate}
\item for $i=2$, consider the following system of $(k+1)(q+1)+1$ equations with $n$ indeterminates $v_1^{q+1},\hdots,v_n^{q+1}$ defined by (\ref{eq:lambdai}).

Then $C_2$ is Hermitian self-orthogonal if and only if (\ref{eq:lambdai}) has a solution. We know that $\lambda_1$ exists by Lemma  \ref{lem:embedding}. The first $k+1$ rows $g_1,\hdots,g_k,g_{k+1}$ of ${\cal G}_2$ generates a Hermitian self-orthogonal $[n+2,k+1]$ code. We now prove that $<g_{k+1},g_{k+2}>_h=0$ and $<g_{k+2},g_{k+2}>_h=-\lambda_2^{q+1}$ for some $\lambda_2\in \F_{q^2}.$
\begin{enumerate}
\item Case $a=4$:\\
Since the dimension $k=3$, $\beta_1^{2(q+1)}=1$ and $\beta_2^{4(q+1)}=1$, the last row of $M$ is
$$
\begin{array}{lll}
(\alpha_1^{4(q+1)},\cdots, \alpha_{n-1}^{4(q+1)},\alpha_n^{4(q+1)})&&\\
&=(u_1^{4(q+1)},\hdots,u_{q+1}^{4(q+1)},\\
&~~~ (u_1\beta_1)^{4(q+1)},\hdots,(u_{q+1}\beta_1)^{4(q+1)},&\\
&~~~(u_1\beta_2)^{4(q+1)},\hdots,(u_{q+1}\beta_2)^{4(q+1)},0)&\\
&=(\underbrace{1,\hdots,1}\limits_{q+1},\underbrace{1,\hdots,1}\limits_{q+1},\underbrace{1,\hdots,1}\limits_{q+1},0),&\\
\end{array}
$$
 and so $v_1^{q+1}\alpha_1^{(k+1)(q+1)}+\cdots+ v_n^{q+1}\alpha_n^{(k+1)(q+1)}=-\lambda_2^{q+1}$ has a solution. Thus $\lambda_2$ is well determined (here $\lambda_2=v_n$). Now, we check that $<g_{k+1},g_{k+2}>_h=0$, that is, 
\begin{equation}
\sum\limits_{i=1}^n\alpha_i^{kq+k+1}v_i^{q+1}=0.
\label{eq:gkandgk1}
\end{equation}
Since $k(q+1)+1=n$, $v_i^{q+1}=\text{Res}_{P_{\alpha_i}}(\omega)$ is a solution of (\ref{eq:gkandgk1}).
\item Case $a=3$:\\
Since the dimension $k=3$, $\beta_1^{2(q+1)}=1$ and $\beta_2^{3(q+1)}=1$, the last row of $M$ is
$$
\begin{array}{lll}
(\alpha_1^{4(q+1)},\cdots, \alpha_{n-1}^{4(q+1)},\alpha_n^{4(q+1)})&&\\
&=(u_1^{4(q+1)},\hdots,u_{q+1}^{4(q+1)},&\\
&~~ ~(u_1\beta_1)^{4(q+1)},\hdots,(u_{q+1}\beta_1)^{4(q+1)},&\\
&~~ ~(u_1\beta_2)^{4(q+1)},\hdots,(u_{q+1}\beta_2)^{4(q+1)},0)&\\
&=(\underbrace{1,\hdots,1}\limits_{q+1},\underbrace{1,\hdots,1}\limits_{q+1},\underbrace{\beta_2^{q+1},\hdots,\beta_2^{q+1}}\limits_{q+1},0),&\\
\end{array}
$$ and so $v_1^{q+1}\alpha_1^{(k+1)(q+1)}+\cdots+ v_n^{q+1}\alpha_n^{(k+1)(q+1)}=-\lambda_2^{q+1}$ has a solution. Thus, $\lambda_2$ is well determined. Now, we check that $<g_{k+1},g_{k+2}>_h=0$, that is, (\ref{eq:gkandgk1}) holds.
Since $k(q+1)+1=n$, $v_i^{q+1}=\text{Res}_{P_{\alpha_i}}(\omega)$ is a solution of (\ref{eq:gkandgk1}).
\end{enumerate}

\item for $i>2$, the Hermitian self-orthogonality follows from the same reasoning as point 1).
\end{enumerate}
\end{proof}


Next, we provide a more general embedding construction of Hermitian self-orthgonal codes from an MDS Hermitian self-orthogonal code.

\begin{thm}\label{thm:embedding-new3}
 Let $q$ be an odd prime power, $D=P_{\alpha_1}+\cdots+P_{\alpha_n}$ and $G=(k-1)O$ be two divisors, $\omega$ be a Weil differential form such that $H=D-G+(\omega)$ and ${\bf v}=(v_1,\hdots,v_n)$ with $v_i\in \F_{q^2}^*.$ Assume that the following conditions hold
\begin{enumerate}
\item $G\le H$,
\item $\text{Res}_{P_{\alpha_j}}(\omega)=v_j^{q+1}$ for $1\le j\le n$,
\item $2\le k\le \lfloor \frac{n+q-1}{q+1} \rfloor$ and $k+q\le n/2$.
\end{enumerate}
 Then the code $C_{{\cal L}_{q}}(D,G;{\bf v})$ can be embedded into a Hermitian self-orthogonal code with parameters $[n+i,k+i,\ge n-k-i+1]_{q^2}$ for $1\le i \le q$.
\end{thm}
\begin{proof} Denote $G'=(k+i-1)O$ and ${\bf w}=(w_i,\hdots,w_n)$, where $w_i=v_i^{\frac{q+1}{2}}$ for $1\le i\le n$. Under the above assumption, the code $C_{{\cal L}_{q}}(D,G;{\bf v})$ is a Hermitian self-orthogonal code with parameters $[n,k,n-k+1]_{q^2}$, and the code $C_{{\cal L}}(D,G';{\bf w})$ is a Euclidean self-orthogonal code with parameters $[n,k+i,n-k-i+1]_{q^2}$ for $1\le i\le q$ since $k+q\le \frac{n}{2}$ and $Res_{P_{\alpha_j}}$ are non-zero square elements in $\F_{q^2}$ for $1\le j\le n$. Let $C_i$ be an $[n+i,k+i]$ code with its generator matrix ${\cal G}_i$  defined as in (\ref{eq:Gi}). We now prove that $C_i$ is Hermitian self-orthogonal. From the Euclidean self-orthogonality of the code $C_{{\cal L}}(D,G';{\bf w})$, we deduce that for $1\le r\le k-1+i, 1\le i\le q$,
\begin{equation}\label{eq:E-orthogonal}
\sum\limits_{j=1}^nw_j^2\alpha_j^{r(k-1+i)}=0.
\end{equation}
By raising both sides of (\ref{eq:E-orthogonal}) to the power $q$, we get 
\begin{equation}
\begin{array}{cc}
\sum\limits_{j=1}^nw_j^{2q}\alpha_j^{rq(k-1+i)}&=0\\
\sum\limits_{j=1}^nv_j^{q+1}\alpha_j^{rq(k-1+i)}&=0.\\
\end{array}
\label{eq:H-orthogonal}
\end{equation}
The last equality of (\ref{eq:H-orthogonal}) holds due to fact that $(v_j^{q+1})_{1\le j\le n}$ are elements in $\F_q^*$. Now from the last equality of (\ref{eq:H-orthogonal}),  we deduce that $<g_j,g_l>_h=0$ for $j\not =l$, where $(g_j)_{1\le j\le k+i}$ are the rows of ${\cal G}_i$. The rest is to check the existence of $(\lambda_j)_{1\le j\le i}$ satisfying
\begin{equation}
\sum\limits_{j=1}^nv_j^{q+1}\alpha_j^{(q+1)(k-1+r)}=-\lambda_r^{q+1}\text{ for } 1\le r\le i.
\label{eq:ex-lambdai}
\end{equation}
Following from the same reasoning as in the proof of \cite[Lemma 4]{SokQSC}, such $(\lambda_r)_{1\le r\le i}$ do exist by taking $\text{Res}_{P_{\alpha_j}}(\omega)=v_j^{q+1}$ for $1\le j\le n$.
\end{proof}


\begin{cor} Let $q$ be an odd prime power, $D=P_{\alpha_1}+\cdots+P_{\alpha_n}$ and $G=(k-1)O$ be two divisors, $\omega$ be a Weil differential form such that $H=D-G+(\omega)$ and ${\bf v}=(v_1,\hdots,v_n)$ with $v_i\in \F_{q^2}^*.$ Assume that the following conditions hold
\begin{enumerate}
\item $G\le H$,
\item $\text{Res}_{P_{\alpha_j}}(\omega)=v_j^{q+1}$ for $1\le j\le n$,
\item $2\le k\le \lfloor \frac{n+q-1}{q+1} \rfloor$ and $k+q\le n/2$.
\end{enumerate}
 Then the code $C_{{\cal L}_{q}}(D,G;{\bf v})$ can be embedded into a Hermitian self-orthogonal code with parameters $[n,2k-1,\ge n-k-q+2]_{q^2}$.
\end{cor}

\begin{proof} From Theorem \ref{thm:embedding-new3}, there exists a Hermitian self-orthogonal code $C_{q-1}$ with parameters $[n+q-1,k+q-1,\ge n-k-q+2]_{q^2}$. From Remark \ref{rem:1}, we get that $\gamma_0=0, \gamma_1=\gamma_q=\lambda_{q-k+1}=0,\hdots, \gamma_{k-1}=\gamma_{q+k-2}=\lambda_{q-1}=0$. By taking the subcode of $C_{q-1}$ whose rows contain the first $k$ rows of ${\cal G}_i$ and the $(k-1)$ consecutive rows with $\lambda_i=0$, we obtain a Hermitian self-orthogonal code with parameters $[n+q-1,2k-1,\ge n-k-q+2]$. Now by puncturing the last $q-1$ zero columns, we get a Hermitian self-orthogonal with parameters $[n,2k-1,\ge n-k-q+2]$, and this completes the proof.
\end{proof}
\begin{exam} We provide two good Hermitian self-orthogonal codes as follows.
\begin{enumerate}
\item Take $q=5$, $k=3$ and $n=19$. Then we can construct a Hermitian self-orthogonal code with parameters $[19,5,\ge 13]_{5^2}$. 
\item Take $q=5$, $k=4$ and $n=25$. Then we can cosntruct a Hermitian self-orthogonal code with parameters $[25,7,\ge 18]_{5^2}$.
\end{enumerate}
More parameters of some good Hermitian self-orthogonal codes, calculated by Magma, are given in Tables \ref{table:hermitian-so-1} and \ref{table:hermitian-so-2}.
\end{exam}

\begin{table}
\centering
\caption{Some Hermitian self-orthogonal codes with good paramters over $\F_{q^2}$, $q=2,3,4$, $n_i=n+i$, $k_i=k+i$, $(n-1)|(q^2-1)$, $k=\lfloor \frac{n+q-1}{q+1}\rfloor$
}
$$
\begin{array}{cccc}
i&[n_i,k_i,d_i]_{q^2}&\text{ Punctured code}&I\\
\hline
\hline
1&[5, 2, 4]_{2^2}&-&\{1\}\\

2&[6, 3, 4]_{2^2}&-&\{1,2\}\\

\hline
\hline

1&[6, 2, 5]_{3^2}&-&\{1\}\\

2&[7, 3, 4]_{3^2}&-&\{1,2\}\\

3&[8, 4, 4]_{3^2}&-&\{1,2,3\}\\

1&[10, 3, 8]_{3^2}&-&\{1\}\\

2&[11, 4, 6]_{3^2}&[10, 4, 6]_{3^2}&\{1\}\\

3&[12, 5, 6]_{3^2}&[11, 5, 6]_{3^2}&\{1,3\}\\

\hline
\hline

1&[7, 2, 6]_{4^2}&-&\{1\}\\

2&[8, 3, 6]_{4^2}&-&\{1,2\}\\

3&[9, 4, 6]_{4^2}&-&\{1,2,3\}\\

4&[10, 5, 6]_{4^2}&-&\{1,2,3,4\}\\

1&[17, 4, 14]_{4^2}&-&\{1\}\\

2&[18, 5, 12]_{4^2}&[17, 5, 12]_{4^2}&\{1\}\\

3&[19, 6, 11]_{4^2}&[17, 6, 11]_{4^2}&\{1\}\\

4&[20, 7, 11]_{4^2}&[18, 7, 11]_{4^2}&\{1,4\}\\

\end{array}
\label{table:hermitian-so-1}
$$
\end{table}

\begin{table}
\centering
\caption{Some Hermitian self-orthogonal codes with good paramters over $\F_{q^2}$, $q=5,7$, $n_i=n+i$, $k_i=k+i$, $(n-1)|(q^2-1)$ or $n=3(q+1)+1$, $k=\lfloor \frac{n+q-1}{q+1}\rfloor$
}
$$
\begin{array}{cccc}
i&[n_i,k_i,d_i]_{q^2}&\text{ Punctured code}&I\\
\hline
\hline
1&[8, 2, 7]_{5^2}&-&\{1\}\\

2&[9, 3, 6]_{5^2}&-&\{1,2\}\\

3&[10, 4, 5]_{5^2}&-&\{1,2,3\}\\

4&[11, 5, 5]_{5^2}&-&\{1,2,3,4\}\\

5&[12, 6, 5]_{5^2}&-&\{1,2,3,4,5\}\\

1&[10, 3, 7]_{5^2}&[9,3,7]_{5^2}&\{\}\\

2&[11, 4, 6]_{5^2}&[9,4,6]_{5^2}&\{\}\\

1&[14, 3, 12]_{5^2}& -&\{1\}\\

2&[15, 4, 10]_{5^2}&[14, 4, 10]_{5^2}&\{1\}\\

3&[16, 5, 10]_{5^2}&[15, 5, 10]_{5^2}&\{1,3\}\\

4&[17, 6, 8]_{5^2}&[15, 6, 8]_{5^2}&\{1,3\}\\

5&[18, 7, 8]_{5^2}&[16, 7, 8]_{5^2}&\{1,3,5\}\\

1&[20, 4,17]_{5^2}&-&\{ 1 \}\\

2&[21, 5, 16]_{5^2}&-&\{1, 2 \}\\

3&[22, 6, 14]_{5^2}&[21, 6, 14]_{5^2}&\{1,2\}\\

4&[23, 7, 13]_{5^2}&[21, 7, 13]_{5^2}&\{1,2\}\\

5&[24, 8, 13]_{5^2}&[22, 8, 13]_{5^2}&\{1,2,5\}\\

1&[26, 5, 22]_{5^2}&  [26, 5, 22]_{5^2}&\{1\}\\

2&[27, 6, 20]_{5^2}&   [26, 6, 20]_{5^2}&\{1\}\\

3&[28, 7, 19]_{5^2}&  [26, 7, 19]_{5^2}&\{1\}\\

4&[29, 8, 18]_{5^2}&  [26,8,18]_{5^2}&\{1\}\\

5& [30,9,18]_{5^2}&  [27,9,18]_{5^2}&\{1,5\}\\

\hline
\hline

1&[14, 3, 11]_{7^2}&[13,3,11]_{7^2}&\{\} \\

2&[15, 4, 11]_{7^2}&[14,4,11]_{7^2}&\{2\}\\

3&[16, 5, 9]_{7^2}&[14,5,9]_{7^2}&\{2\}\\

1&[18,3,16]_{7^2}&  [18,3,16]_{7^2}&\{1\}\\

2&[19,4,14]_{7^2}&  [18,4,14]_{7^2}&\{1\}\\

3&[20,5,14]_{7^2}&  [19,5,14]_{7^2}&\{1,3\}\\

4&[21,6,12]_{7^2}&  [19,6,12]_{7^2}&\{1,3\}\\

5&[22,7,12]_{7^2}&  [20,7,12]_{7^2}& \{1,3,5\}\\

6&[23,8,10]_{7^2}& [20,8,10]_{7^2}&\{1,3,5\}\\

7&[24,9,10]_{7^2}& [21,9,10]_{7^2}&\{1,3,5,7\}\\

1&[26, 4, 23] _{7^2}&-&\{ 1\}\\

2&[27, 5, 21]_{7^2}&[26, 5, 21]_{7^2}&\{1\}\\

3&[28, 6, 20]_{7^2}&[26, 6, 20]_{7^2}&\{1\}\\

4&[29, 7, 20]_{7^2}&[27, 7, 20]_{7^2}&\{1,4\}\\

5&[30, 8, 18]_{7^2}&[27, 8, 18]_{7^2}&\{1,4\}\\

6&[31,9,17]_{7^2}&[27, 9, 17]_{7^2}&\{1,4\} \\

7&[32,10,\ge 16]_{7^2}&[28, 10,\ge 16]_{7^2}&\{1,4,7\}\\

\end{array}
\label{table:hermitian-so-2}
$$
\end{table}

\section{Hermitian hull codes from algebraic curves}\label{section:algebraic curves}
In this section, we construct Hermitian hull codes from algebraic curves. In the first subsection, we begin our construction by considering the projective lines, the projective curves of genus zero, and hence from the Riemann-Roch's theorem, we obtain MDS Hermitian $\ell$-$dim$ hull codes. In the second subsection, we consider Hermitian $\ell$-$dim$ hull codes from algebraic curves of genus greater than or equal to one.
\subsection {Hermitian hull codes from projective lines}
In this subsection, we consider Hermitian hull codes over $\F_{q^2}$ from the projective lines ${\mathbb P}^{1}=\{(\alpha:1)|\alpha\in \F_{q^2}\}\cup \{(1:0)\}$. There are $q^2+1$ distinct points in ${\mathbb P}^{1}$, and thus are $q^2+1$ distinct places of $\F_{q^2}({\mathbb P}^1)$.

\begin{prop}\label{thm:MDS-existence}
There exist MDS Hermitian $\ell$-$dim$ hull codes with parameters $[n,k]_{q^2}$ and $[n,n-k]_{q^2}$ for $0\le \ell \le k$ if one of the following conditions holds:
\begin{enumerate}
    \item $n=q^2+1$, $k=q$; 
 $n=q^2+1$, $ k\le q$, $q$ even, $d$ odd;
  $n=q^2+1$, $k\le q$, $q\equiv 1\mod{4}$, $d$ even;
       \item $n= q^2$, $k\leq q-1$;
  \item $n=(q^2+1)/2$, $q/2< k\le q-1$, $q$ odd;
\item $n=q^2+1$, $k\le q$, $k\not=q-1$;
  $n=r(q-1)+1$, $k\le (q+r-1)/2$, $q\equiv r-1\mod{2r}$;$n=(q^2+2)/3$, $3|(q+1)$, $k\le (2q-1)/3$;
\item $n=tq$, $1\le t\le q$, $ k\le \lfloor\frac{tq+q-1}{q+1} \rfloor$;$n=t(q+1)+2$, $1\le t\le q-1$, $k\le t +1$, $(p,t,k)\not=(2,q-1,q-1)$;
\item $ k\le \lfloor\frac{n+q-1}{q+1} \rfloor$,
\begin{enumerate}
\item $(n-1)|(q^2-1)$;
\item $n=2N+1$, $m=2s$, $q_0=p^s$ odd, $N= \frac{q-1}{p^r+1}$ even, $N<q_0+1$, $r|\frac{m}{2}$;
\item $n=(t+1)N+i$, $1\le i\le 2$, $N|(q^2-1)$, $n_2=\frac{N}{\gcd (N,q+1)}$, $1\le t\le \frac{q-1}{n_2}-1$.
\end{enumerate}
\end{enumerate}
\end{prop}
\begin{proof}First note that if $C$ is a Hermitian $\ell$-$dim$ hull $[n,k]$ code, then its Hermitian dual $C^{\perp_h}$ is a Hermitian $\ell$-$dim$ hull $[n,n-k]$ code. Moreover, the dual of an MDS code is again an MDS code. It is enough to prove the existence of an MDS Hermitian self-orthogonal code with each of the above conditions in the theorem, and thus the result follows from Lemma \ref{lem:char2}. The existence of such codes can be obtained as follows:
\begin{enumerate}
  \item $n=q^2+1$, $k=q$ from \cite{Li Xin Wan2}; 
\item $n=q^2+1$, $ k\le q$, $q$ even, $d$ odd  from \cite{Gu11};
\item  $n=q^2+1$, $k\le q$, $q\equiv 1\mod{4}$, $d$ even from \cite{KZ12};
       \item $n= q^2$, $k\leq q-1$ from \cite{Gra Bet2} and \cite{Li Xin Wan2,LinLingLuoXin};
  \item $n=(q^2+1)/2$, $q/2< k\le q-1$, $q$ odd from \cite{KZ12};
\item $n=q^2+1$, $k\le q$, $k\not=q-1$;
  $n=r(q-1)+1$, $k\le (q+r-1)/2$, $q\equiv r-1\mod{2r}$;$n=(q^2+2)/3$, $3|(q+1)$, $k\le (2q-1)/3$ from \cite{JinXing14};
\item $n=tq$, $1\le t\le q$, $ k\le \lfloor\frac{tq+q-1}{q+1} \rfloor$;$n=t(q+1)+2$, $1\le t\le q-1$, $k\le t +1$, $(p,t,k)\not=(2,q-1,q-1)$ from \cite{FangFu};
\item $(n-1)|(q^2-1)$;$n=2N+1$, $m=2s$, $q_0=p^s$ odd, $N= \frac{q-1}{p^r+1}$ even, $N<q_0+1$, $r|\frac{m}{2}$;
$n=(t+1)N+i$, $1\le i\le 2$, $N|(q^2-1)$, $n_2=\frac{N}{\gcd (N,q+1)}$, $1\le t\le \frac{q-1}{n_2}-1$ from \cite{SokQSC}.
\end{enumerate}
\end{proof}

By applying some propagation rules to an MDS Hermitian self-orthogonal code, we obtain some MDS Hermitian self-orthogonal codes with smaller lengths and smaller dimensions.
Since some families in Proposition \ref{thm:MDS-existence} are contained in other families, we only state the larger families of Hermitian $\ell$-$dim$ hull codes in the following theorem. 
\begin{thm}\label{thm:large-family} There exist MDS Hermitian $\ell$-$dim$ hull codes with parameters 
$[n-s,k-s,n-k+1-s]_{q^2}$ and $[n-s,n-k-s,k+1-s]_{q^2}$ for $1\le s\le k-1$ and $0\le \ell \le k-s$
if one of the following conditions holds
\begin{enumerate}
\item $n=q^2+1$, $k\le q$, $k\not=q-1$;
 $n=r(q-1)+1$, $k\le (q+r-1)/2$, $q\equiv r-1\mod{2r}$;$n=(q^2+2)/3$, $3|(q+1)$, $k\le (2q-1)/3$;
\item $n=tq$, $1\le t\le q$, $ k\le \lfloor\frac{tq+q-1}{q+1} \rfloor$;$n=t(q+1)+2$, $1\le t\le q-1$, $k\le t +1$, $(p,t,k)\not=(2,q-1,q-1)$;
\item $(n-1)|(q^2-1)$, $ k\le \lfloor\frac{n+q-1}{q+1} \rfloor$;
$n=(t+1)N+i$, $1\le i\le 2$, $N|(q^2-1)$, $n_2=\frac{N}{\gcd (N,q+1)}$, $1\le t\le \frac{q-1}{n_2}-1$, $ k\le \lfloor\frac{n+q-1}{q+1} \rfloor$.
\end{enumerate}
\end{thm}
\begin{proof} Assume that there exists a Hermitian self-orthogonal code $C$ with parameters $[n,k,d]_{q^2}$. Then up to equivalence, we can write the generator matrix $G$ of $C$ as $G=(I_k|A).$ Deleting the last $s$ consecutive rows and the  $s$ consecutive zero-columns of $G$, we obtain a matrix $G'$ which again generates a Hermitian self-orthogonal code $C'$. Moreover, if $C$ is MDS, then so is $C'.$ This completes the proof.
\end{proof}

It should be noted for $n$ sufficiently large enough, say $n > k+q$ (see \cite{DodLan}), the dual of an AMDS code is again an AMDS code.
\begin{thm}\label{thm:family-AMDS} Let $q=p^m$, $(n-1)|(q^2-1)$ and $1\le k\le \lfloor \frac{n+q-1}{q+1} \rfloor$.
\begin{enumerate}
\item If $(n-1)|(k+1)(q+1)$, then there exists a Hermitian $\ell$-$dim$ hull code  parameters $[n+1,k+2,n-k-1]_{q^2}$ for $0\le \ell \le k+2$
\item   If $(n-1)|k(q+1)$, then there exists a Hermitian $\ell$-$dim$ hull code  with parameters $[n+i-1,k+i,n-k-i+1]_{q^2}$ for $0\le \ell \le k+i$.
\end{enumerate}

\end{thm}
\begin{proof} The proof follows from Theorem \ref{thm:embedding-new}.
\end{proof}

\begin{cor} 
\begin{enumerate}
\item There exist AMDS Hermitian $\ell$-$dim$ hull codes with parameters $[q^2+1,k]_{q^2}$ and $[q^2+1,q^2+1-k]_{q^2}$ for $0\le \ell \le k$, $k=q+1,\hdots,2q-2$;
\item There exist AMDS Hermitian $\ell$-$dim$ hull codes with parameters $[2(q+1)+2,k+2]_{q^2}$ and $[2(q+1)+2,2(q+1)-k]_{q^2}$ for $0\le \ell \le k+2$;
\item There exist Hermitian $\ell$-$dim$ hull codes with parameters $[2(q+1)+i,k+i, \ge n-k-i+1]_{q^2}$ for $0\le \ell \le k+i$ and $2\le i\le q$.
\end{enumerate}
\end{cor}
\begin{proof}The proof follows from Corollary \ref{cor:explicit-amds}.
\end{proof}
\begin{thm}\label{thm:family-new-mds} Let $q=p^m$, $(n-1)|(q^2-1)$ and $k=\lfloor \frac{n+q-1}{q+1}\rfloor.$  Then
\begin{enumerate}
\item if $(n-1)|k(q+1)$, then there exist MDS Hermitian $(k+i-\sharp I)$-$dim$ hull codes with parameters $[n,k+i]_{q^2}$ and $[n,n-k-i]_{q^2}$ for any $1\le i\le q$;
\item if $(n-1)|(k+1)(q+1)$, then there exist MDS Hermitian $(k+1)$-$dim$ hull codes with parameters $[n,k+2]_{q^2}$ and $[n,n-k-2]_{q^2}$.
\end{enumerate}
\end{thm}
\begin{proof} We prove the theorem using similar technique to that in the proof of Theorem \ref{thm:embedding-new}.
\begin{enumerate}

\item  If $(n-1)|k(q+1)$, then from Theorem \ref{thm:embedding-new}, there exists a Hermitian self-orthogonal $[n+i,k+i]_{q^2}$ code. Consider the code $C_i$ with its generator matrix ${\cal G}_i$ given by (\ref{eq:Gi}).
By puncturing the last $i$ coordinates of ${\cal G}_i$, one obtains a generator matrix ${\cal G}_{i_0}$ satisfying $\text{Rank}({\cal G}_{i_0}{\cal G}_{i_0}^\dag)=k+i-\sharp I$. It is not difficult to check that ${\cal G}_{i_0}$ generates an MDS Hermitian $(k+i-\sharp I)$-$dim$ hull code.
\item  If $(n-1)|(k+1)(q+1)$, then there exist $\lambda_1\not=0$ and $\lambda_2=0$, and the result follows from point 1).
\end{enumerate}

\end{proof}

\begin{exam} One can easily obtain an MDS Hermitian $7$-$dim$ hull $[25,9,17]_{5^2}$ code and an MDS Hermitian $6$-$dim$ hull $[25,8,18]_{5^2}$ code from the generator matrix of Example \ref{ex:1}, as well as an MDS Hermitian $5$-$dim$ hull $[17,9,9]_{5^2}$ code and an MDS Hermitian $5$-$dim$ hull $[17,8,10]_{7^2}$ code from the generator matrix of Example \ref{ex:2}. These four codes have their dimension $k>\frac{n+q-1}{q+1}$, and thus their parameters are never reachable by \cite{FangFuLiZhu}.
\end{exam}

\begin{cor} Let $q=p^m$ and $k=\lfloor \frac{n+q-1}{q+1}\rfloor$. Then
\begin{enumerate}
\item for $n=q^2$, there exist MDS Hermitian $(k+i-1)$-$dim$ hull codes with parameters $[q^2,k+i]_{q^2}$ and $[q^2,q^2-k-i]_{q^2}$ for $1\le i\le q$;
\item $n=2(q+1)+1$, there exist MDS Hermitian $(k+i-\lfloor \frac{i-1}{2} \rfloor -1)$-$dim$ hull codes with parameters $[2(q+1)+1,k+i]_{q^2}$ and $[2(q+1)+1,2(q+1)+1-k-i]_{q^2}$ for $1\le i\le q$.
\end{enumerate}
\end{cor}

\begin{proof}
\begin{enumerate}
\item Take $n-1=q^2-1$. Then we get that $\lambda_1\not=0$ and $\lambda_2=\cdots=\lambda_{q-1}=0$. The result follows from Theorem \ref{thm:family-new-mds} 1).
\item Take $n-1=2(q+1)$. Then we get that $\lambda_1=\lambda_3=\cdots=\lambda_{q-2}\not=0$ and $\lambda_2=\cdots=\lambda_{q-1}=0$. The result follows from Theorem \ref{thm:family-new-mds} 1).
\end{enumerate}
\end{proof}

\begin{thm}\label{thm:family-AMDS2} Let $q$ be an odd prime power, $N|(q^2-1)$, $n_2=\frac{N}{\gcd (N,q+1)}$, $1\le t\le \frac{q-1}{n_2}-1$. Put $n=(t+1)N+1$, $ k= \lfloor\frac{n+q-1}{q+1} \rfloor$, $k_i=k+i$ and $n_i=n+i$. Then 
\begin{enumerate}
\item there exist AMDS Hermitian $\ell$-$dim$ hull codes with parameters $[n+2,k+2]_{q^2}$ and $[n+2,n-k]_{q^2}$ for $0\le \ell \le k+2$;
\item  there exist Hermitian $\ell$-$dim$ hull codes with parameters $[n+i,k+i,\ge n-k-i+1]_{q^2}$ code for $0\le \ell \le k+i$ for $1\le i \le q$;
\item there exist Hermitian $\ell$-$dim$ hull codes with parameters $[n,2k-1,\ge n-k-q+2]_{q^2}$ for $0\le \ell \le 2k-1$ ;
\item  there exist MDS Hermitian $(k+i-\sharp I)$-$dim$ hull codes with parameters $[n,k+i]_{q^2}$ and $[n,n-k-i]_{q^2}$ for $1\le i \le q$.
\end{enumerate}
\end{thm}
\begin{proof}It was proved in \cite[Construction 7]{SokQSC} that there exists an MDS Hermitian self-orthogonal code with parameters $[n,k]_{q^2}$. Now, the result follows from the same reasoning as in the proof of Theorem \ref{thm:family-new-mds}.
\end{proof}
In the following, we give another family of MDS Hermitian $\ell$-$dim$ hull codes. 

For sake of stating our result, let us fix the following notation and setting. Let $q=p^m$ be a prime power and $n|(q^2-1)$. Write
\begin{equation}
n_{1}=\gcd (n, q-1) \text{ and }n_{2}=\frac{n}{n_1}. 
\label{eq:n1n2}
\end{equation}
Then the first equality of (\ref{eq:n1n2}) implies that $n_2$ and $\frac{q-1}{n_1}$ are coprime, and thus we get that $n_2|(q+1)\frac{q-1}{n_1}$, and so $n_2|(q+1)$.
Let $U_n$ and $V_n$ be two subgroups of $\mathbb{F}_{q^{2}}^{*}$ generated by $w^{\frac{q^{2}-1}{n}}$ and $w^{\frac{q-1}{n_{1}}}$, respectively, where $w$ is a primitive element of $\mathbb{F}_{q^{2}}$. It can be easily seen that $|U_n|=n$ and $|V_n|=(q+1)n_{1}$. Now $\frac{q^{2}-1}{n}=\frac{q-1}{n_{1}}\cdot\frac{q+1}{n_{2}}$ implies that $ \frac{q-1}{n_{1}} \mid \frac{q^{2}-1}{n}$, and we deduce that $U_n$ is a subgroup of $V_n$. Write $U_n=\{u_1,\hdots,u_n\}$. 

\begin{prop}\label{const:multicoset2} Let $q=p^m$ be a prime power, $n|(q^2-1)$ and $n_2=\frac{n}{\gcd (n,q-1)}$. Put $N=(t+1)n+1$ for $1\le t\le \frac{q+1}{n_2}-1$. Assume that there exists $t$ such that for all $\alpha_1,\hdots,\alpha_t\in V_n\backslash U_n$, we have $\alpha_i^n=\beta_i^{q+1}$ for some $\beta_i\in \F_{q^2}$. Then for $1\le k\le \lfloor \frac{N+q-1}{q+1} \rfloor$, there exist MDS Hermitian $\ell$-$dim$ hull codes with parameters $[N,k]_{q^2}$, $[N+1,k+1]_{q^2}$, $[N,N-k]_{q^2}$ and $[N+1,N-k]_{q^2}$ for $0\le \ell \le k$. 
\end{prop}

\begin{proof}
Let $\alpha_{1}U_n, \ldots , \alpha_{\frac{q+1}{n_{2}}-1}U_n $ be all the distinct cosets of $V_n$ different from $U_n$.

For $1 \leq t \leq \frac{q+1}{n_{2}}-1$, put $U=U_n \bigcup^{t}\limits_{j=1}\alpha_{j}U_n\cup \{0\}$, say $U=\{a_1,\hdots,a_{(t+1)n+1}\}$, and  write $$h(x)=\prod\limits_{\alpha\in U}(x-\alpha).$$
Then the derivative of $h(x)$ is given by 
\begin{equation*}
\begin{array}{ll}
h'(x)=&((n+1)x^n-1)\prod\limits_{i=1}^t(x^n-\alpha_i^n)\\
&+nx^n(x^n-1)\left(\sum\limits_{i=1}^t\prod\limits_{j=1,j\not=i}^t(x^n-\alpha_j^n)\right).\\
\end{array}
\end{equation*}

For $1\le j\le t,1\le s\le n$, we have 
\begin{equation} h'(\alpha_ju_s)=n\alpha_j^n(1-\alpha_j^n)\prod\limits_{i=1,i\not=j}^t(\alpha_j^n-\alpha_i^n).
\label{eq:derivative}
\end{equation}

Now, we check the values of the derivative $h'$ given by (\ref{eq:derivative}).

Since $\alpha_{j}$ is in $V_n$, we can write $\alpha_{j}=w^{e_j\frac{q-1}{n_{1}}}$ for some positive integer $e_j$. Thus,
$\alpha_{j}^{n}=w^{e_jn\frac{q-1}{n_{1}}}=w^{e_jn_{2}(q-1)}$, and $\alpha_{j}^{n(q+1)}=1$, that is, $\alpha_{j}^{nq}=\alpha_j^{-n}.$ Now, the latter equality implies that 
$(\alpha_i^n-\alpha_j^n)^q=\frac{\alpha_j^n-\alpha_i^n}{\alpha_i^n\alpha_j^n},$
and thus 
\begin{equation}
\begin{array}{ll}
(\alpha_i^n-\alpha_j^n)^{q-1}&=\frac{-1}  {(\alpha_i\alpha_j)^n}
=\frac{-1}  {(\beta_i\beta_j)^{q+1}}
\end{array}
\label{eq:betaij}
\end{equation}
By raising the last equality of (\ref{eq:betaij}) to the power $q+1$, we get $(\beta_i\beta_j)^{q+1}=1$, and this implies that
\begin{equation}
\begin{array}{ll}
\alpha_i^n-\alpha_j^n&=w^{\frac{q+1}{2}}\text{ which is independent of $i,j$.}
\end{array}
\label{eq:alphaij}
\end{equation}

Set $G=(k-1)O$, $D=(h)_0$, $\omega=\frac{dx}{h(x)}$ and $H=((t+1)n-k-1)O$. Thus by Lemma \ref{thm:key}, the code $C_{{\cal L}_{q}}(D,G;{\bf v})$, where ${\bf v}=(v_1,\hdots,v_{(t+1)n+1})$ and $v_i^{q+1}=1/\beta_i$, is Hermitian self-orthogonal over $\F_{q^2}$.
Hence, for $1\le k\le \frac{(t+1)n+q}{q-1}$, there exists an MDS Hermitian self-orthogonal $[(t+1)n+1,k]_{q^2}$ code. By applying the embedding Lemma \ref{lem:embedding}, we get an MDS Hermitian self-orthogonal $[(t+1)n+2,k+1]_{q^2}$ code. Hence, the result follows from Lemma \ref{lem:char2}.
\end{proof}

\begin{cor}
Let $q=p^m$ be a prime power and $N=2(q+1)+1$. Then for $1\le k\le \lfloor \frac{N+q-1}{q+1} \rfloor$, there exist MDS Hermitian $\ell$-$dim$ hull codes with parameters $[N,k]_{q^2}$, $[N+1,k+1]_{q^2}$, $[N,N-k]_{q^2}$ and $[N+1,N-k]_{q^2}$ for $0\le \ell \le k$. 
\end{cor}

\begin{proof} Take $n=q+1$. Then $n_1=\gcd(n,q-1)=2$ and thus $n_2=\frac{q+1}{2}$. With similar argument as in the proof of Theorem \ref{const:multicoset2}, we get 
\begin{equation*}
\begin{array}{ll}
(1-\alpha_1^{q+1})&=\frac{w^{\frac{q+1}{2}}}  {w^{e_1n_2}}=\frac{w^{\frac{q+1}{2}}}  {w^{e_1\frac{q+1}{2}}}\\
\end{array}
\end{equation*}
If $e_1$ was even, then $\alpha_1^n=w^{(q-1)n_2e_1}=w^{(q^2-1)\frac{e_1}{2}}=1$, and thus $\alpha_1\in U_n$ which is a contradiction. Hence, $e_1$ is odd, and the proof is completed.
\end{proof}

\subsection{Hermitian $\ell$-$dim$ hull codes from elliptic curves, hyper-elliptic curves and Hermitian curves}
In the previous subsection, we have already seen that Hermitian self-orthogonal codes, constructed from the projective lines, have their lengths upper bounded  by $q^2+1$ for MDS case ($q>2$), and by $q^2+q+1$ for non MDS case, and thus the code lengths of the corresponding Hermitian $\ell$-$dim$ hull codes can not go beyond $q^2+q+1$. In this subsection, we consider Hermitian $\ell$-$dim$ hull codes from some algbraic curves over $\F_q$ which have many $\F_q$-rational points. These curves could produce long codes with large minimum distance.  By the famous Hasse-Weil bound, the bound of the number of $\mathbb F_q$-rational points of a smooth projective curve $\cal X$ defined over $\mathbb F_q$ is given as follows:
 $$|\sharp{\cal X}(\mathbb F_q)-(q+1)|\le 2 g \sqrt{q},$$
where $g$ is the genus of ${\cal X}$. When the bound is attained, the curves are called maximal. We are considering some  well-known maximal curves, for example elliptic curves, hyper-elliptic curves and Hermitian curves and employ them to construct Hermitian hull codes from AG codes.

First we consider elliptic curves in even characteristic.
Let $q = 2^m$ and an elliptic curve be defined by the equation
\begin{equation}
{\cal E}_{a,b,c}:~y^2 + ay = x^3 + bx + c,
\label{eq:elliptic-curve}
\end{equation}
where $a, b, c \in\F_q .$ Let denote the number of rational points on ${\cal E}_{a,b,c}$ by ${N}_{a,b,c}$  Let $S$ be the set of $x$-components of the affine points of ${\cal E}_{a,b,c}$ over $\F_q $, that is, 
\begin{equation}
S_{a,b,c} := \{\alpha\in \F_q| \exists \beta \in \F_q\text{ such that } \beta^2+a\beta=\alpha^2+b\alpha+c\}.
\label{eq:set-elliptic}
\end{equation}

It should be noted that any $\alpha\in S_{1,b,c}$ gives rise to two points on ${\cal E}_{a,b,c}$ with $x$-component $\alpha$, and we denote them by $P^{(1)}_\alpha$ and $P^{(2)}_\alpha.$ 
The numbers ${N}_{1,b,c}$ of $\F_q$-rational points of elliptic curves ${\cal E}_{1,b,c}$ over $\F_q$ with $q$ being a square are given in Table \ref{table:size-elliptic}.
\begin{table}
\caption{Numbers of rational points of elliptic curves over $\F_q$ with $q=2^m$ being a square}\label{table:size-elliptic}
\begin{center}
\begin{tabular}{c|c|c}
\text{Elliptic curve }${\cal E}_{1,b,c}$&$m$&${N}_{1,b,c}$\\
\hline
$y^2 + y = x^3$&$m\equiv 0 \pmod 4$&$q+1-2\sqrt{q}$\\
&$m\equiv 2 \pmod 4$&$q+1+2\sqrt{q}$\\
\hline
$y ^2 + y = x ^3 + bx ( T r_1^m ( b) = 1 )$&$m\text{ even }$&$q+1$\\
\hline
\multirow{2}{13em}{$y ^2 + y = x ^3 + c ~( T r_1^ m ( c ) = 1 )$}&$m\equiv 0 \pmod 4$&$q+1+2\sqrt{q}$\\
&$m\equiv 2 \pmod 4$&$q+1-2\sqrt{q}$\\
\end{tabular}
\end{center}
\end{table}

We can now state an existence of AMDS Hermitian $\ell$-$dim$ hull codes from the elliptic curves as follows.
\begin{prop}\label{thm:elliptic}
Let $q=2^s, m=2s$ and $2k+2\le n\le \lfloor q+\sqrt{2q}-5\rfloor$. Then there exist AMDS Hermitian $\ell$-$dim$ hull codes with parameters $[n,k]_{q^2}$ and $[n,n-k]_{q^2}$ for $0\le \ell \le k.$
\end{prop}
\begin{proof} The existence of a Hermitian self-orthogonal $[n,k]_{q^2}$ code follows from \cite{JinXing12}. Hence, the result follows from Lemma \ref{lem:char2}.
\end{proof}

Next, we consider the affine hyper-elliptic curve over $\mathbb F_{q^2}$ with $q=2^m$ which is defined by
\begin{equation}
{\cal C}: ~y^2+y=x^{q+1}.
\label{eq:hyper-elliptic}
\end{equation}
Note that for any $\alpha\in \mathbb F_{q^2}$, there are two rational points on ${\cal C}$ with $x$-component $\alpha$, and we denote them by $P_{\alpha}^{(1)}, P_{\alpha}^{(2)}$. 
The set ${\cal C}(\mathbb F_{q^2})$ of all rational points of ${\cal C}$ is equal to $\{P_{\alpha}^{(1)}| \alpha \in \mathbb F_{q^2}\}\cup \{P_{\alpha}^{(2)}|\alpha \in \mathbb F_{q^2}\} \cup \{ O\}$. This curve has genus $g=\frac{q}{2}$. It can be easily checked that the curve $\cal C$ has $1+2q^2=1+q^2+2\frac{q}{2}q$ points, and thus it is a maximal curve.

The constructions of Hermitian $\ell$-$dim$ hull codes from the hyper-elliptic curves are given as follows.

\begin{prop}\label{thm:hyper-elliptic}

Let $q=2^m, m\ge 2$ and $q\le k\le \lfloor \frac{2q^2+2q-1}{q+1} \rfloor$. Put $G=(k-1)O$ and $D=\sum\limits_{\alpha\in \F_{q^2}}(P_{\alpha}^{(1)}+P_{\alpha}^{(2)}).$ Then for $0\le \ell \le k-q/2$, the code $C_{{\cal L}_q}(D,G)$ gives rise to a Hermitian $\ell$-$dim$ hull $[2q^2,k-q/2,\ge 2q^2-k+1]_{q^2}$ code whose Hermitian dual has parameters $[2q^2,2q^2+q/2-k,\ge k+1-q]_{q^2}$.
\end{prop}

\begin{proof} From \cite{SokQSC}, the code $C_{{\cal L}_q}(D,G)$ is Hermitian self-orthogonal with parameters $[2q^2,k-q/2,\ge 2q^2-k-\frac{q}{2}+1]_{q^2}$. The result follows from Lemma \ref{lem:char2}.
\end{proof}

\begin{prop}\label{thm:hyper-elliptic2}
Let $q=2^m, m\ge 2$, $(n-1)|(q-1)$ and $q\le k\le \lfloor \frac{2n+2q-1}{q+1} \rfloor$. Put $U=\{\alpha\in \F_{q^2}|\alpha^n=\alpha \}$, $G=(k-1)O$ and $D=\sum\limits_{\alpha\in U}(P_{\alpha}^{(1)}+P_{\alpha}^{(2)}).$ Then for $0\le \ell \le k-q/2$, the code $C_{{\cal L}_q}(D,G)$ gives rise to a Hermitian $\ell$-$dim$ hull $[2n,k-q/2,\ge 2n-k+1]_{q^2}$ code whose Hermitian dual has parameters $[2n,2n+q/2-k,\ge k+1-q]_{q^2}$.
\end{prop}
\begin{proof} From \cite{SokQSC}, the code $C_{{\cal L}_q}(D,G)$ is Hermitian self-orthogonal with parameters $[2n,k-q/2,\ge 2n-k+1]_{q^2}$. Hence, the result follows from Lemma \ref{lem:char2}.
\end{proof}

In the following, we consider the affine Hermitian curve over $\mathbb F_{q^2}$ with $q=p^m$ which is defined by
\begin{equation}
{\cal H}: ~y^q+y=x^{q+1}.
\label{eq:hermitian}
\end{equation}
For any $\alpha\in \mathbb F_{q^2}$, there are $q$ rational points on $\cal H$ with $x$-component $\alpha$, and we denote them by $P_{\alpha}^{(1)}, P_{\alpha}^{(2)},\hdots, P_{\alpha}^{(q)}$.
The set ${\cal H}(\mathbb F_{q^2})$ of all rational points of ${\cal H}$ is equal to $\{P_{\alpha}^{(1)}|\alpha \in \mathbb F_{q^2}\}\cup \{P_{\alpha}^{(2)}| \alpha \in \mathbb F_{q^2}\} \cup \cdots \cup \{P_{\alpha}^{(q)}|\alpha \in \mathbb F_{q^2}\}\cup \{O\}$. The curve $\cal H$ has genus $g=\frac{q(q-1)}{2}$ and has $1+q^3=1+q^2+2\frac{q(q-1)}{2}q$ points. Hence, it is a maximal curve.

\begin{prop}\label{thm:hermitian-1}
Let $q=p^m\ge 4$ and $q(q-1)\le k\le \lfloor \frac{q^3+q^2-1}{q+1} \rfloor$. Put ${\bf v}=(1/\sqrt[q+1]{-1},\hdots, 1/\sqrt[q+1]{-1})$,$G=(k-1)O$ and $D=\sum\limits_{\alpha\in \F_{q^2}}(P_{\alpha}^{(1)}+P_{\alpha}^{(2)})$. Then for $0\le \ell \le k-\frac{q(q-1)}{2}$, the code $C_{{\cal L}_q}(D,G;{\bf v})$ gives rise to a Hermitian $\ell$-$dim$ hull $[q^3,k-\frac{q(q-1)}{2},\ge q^3-k+1 ]_{q^2}$ code whose Hermitian dual has parameters $[q^3,q^3+q(q-1)/2-k,\ge k+1-q(q-1)]_{q^2}$.
\end{prop}
\begin{proof} From \cite{SokQSC}, the code $C_{{\cal L}_q}(D,G;{\bf v})$ is Hermitian self-orthogonal with parameters $[q^3,k-\frac{q(q-1)}{2},\ge q^3-k+1 ]_{q^2}$. Hence, the result follows from Lemma \ref{lem:char2}.
\end{proof}

\begin{prop} \label{thm:hermitian-2} Let $q=p^m\ge 4$ and $q(q-1) \le k\le \lfloor \frac{Nq+q^2-1}{q+1} \rfloor$. 
Then for $0\le \ell \le k-\frac{q(q-1)}{2}$, there exists a Hermitian $\ell$-$dim$ hull $[Nq,k-\frac{q(q-1)}{2},\ge Nq-k+1 ]_{q^2}$ code whose Hermitian dual has parameters $[Nq,Nq+q(q-1)/2-k,\ge k+1-q(q-1)]_{q^2}$
if one of the following condition holds
\begin{enumerate}[i)]
\item $n=q^2+1$, $k\le q$, $k\not=q-1$;
 $n=r(q-1)+1$, $k\le (q+r-1)/2$, $q\equiv r-1\mod{2r}$;$n=(q^2+2)/3$, $3|(q+1)$, $k\le (2q-1)/3$;
\item $n=tq$, $1\le t\le q$, $ k\le \lfloor\frac{tq+q-1}{q+1} \rfloor$;$n=t(q+1)+2$, $1\le t\le q-1$, $k\le t +1$, $(p,t,k)\not=(2,q-1,q-1)$;
\item $(n-1)|(q^2-1)$,$ k\le \lfloor\frac{n+q-1}{q+1} \rfloor$;
$n=(t+1)N+i$, $1\le i\le 2$, $N|(q^2-1)$, $n_2=\frac{N}{\gcd (N,q+1)}$, $1\le t\le \frac{q-1}{n_2}-1$, $ k\le \lfloor\frac{n+q-1}{q+1} \rfloor$.
\end{enumerate} 
\end{prop}

\begin{proof} The existence of a Hermitian self-orthogonal $[Nq,k-\frac{q(q-1)}{2},\ge Nq-k+1 ]_{q^2}$ code follows from \cite{SokQSC}. Hence, the result follows from Lemma \ref{lem:char2}.
\end{proof}

\begin{thm}\label{thm:hermitian-3}Let $q=p^m$ be an odd prime power, $1<t< q+1$ and $t|(q+1)$. Put $g=\frac{(q-1)(q+1-t)}{2t}$ and $s={q}\left(\frac{\lfloor q^2-2\rfloor}{t}+2\right)$. Then for $2g \le k\le \lfloor\frac{s+q-1}{q+1}\rfloor$ and $0 \le \ell \le k-g$, there exist $\ell$-$dim$ hull codes with parameters $[s,k-g,\ge s-k+1]_{q^2}$ and $[s,s-k+g,\ge k-2g+1]_{q^2}$.
\end{thm}
\begin{proof}
Consider an algebraic curve defined by
$${\cal X}: y^{q}+y=x^{\frac{{q}+1}{t}}.$$
The curve has genus $g=\frac{(q-1)(q+1-t)}{2t}$.
Put $$U=\{\alpha\in \F_{q^2}|\exists \beta\in \F_q\text{ such that }\beta^{q}+\beta=\alpha^{\frac{q+1}{t}}\}.$$ 
The set $U$ is the set of $x$-component solutions to the Hermitian curve (\ref{eq:hermitian}) whose elements are $t$-th power elements in $\F_{q^2}.$ There are such $\left(\frac{\lfloor q^2-2\rfloor}{t}+2\right)$ elements in $\F_{q^2}$ which give rise to $s={q}\left(\frac{\lfloor q^2-2\rfloor}{t}+2\right)$ rational places. Write
$$h(x)=\prod\limits_{\alpha\in U}(x-\alpha)\text{ and } \omega=\frac{dx}{h(x)}.$$
Then $h(x)=x^n-x$, where $n=\frac{\lfloor q^2-2\rfloor}{t}+2,$ and thus $h'(x)=nx^{n-1}-1$. Obviously, $h'(0)=w^{\frac{q^2-1}{2}}$, where $w$ is a primitive element of $\F_{q^2}$. Since $q$ is a square, we have that for any $\alpha\in U \backslash \{0\}$, $h'(\alpha)=n-1=\beta^{q+1}$ for some $\beta\in \F_{q^2}$.
Put $D=\sum\limits_{\alpha\in U}\left(P_\alpha^{(1)}+\cdots+P_\alpha^{(q)}\right)=P_1+\cdots+P_s$, $G=(k-1)O$. By Lemma \ref{thm:key}, the constructed code $C_{{\cal L}_q}(D,G;{\bf v})$ is Hermitian self-orthogonal, where ${\bf v}=(\underbrace{1/w^{\frac{q-1}{2}},\hdots,1/w^{\frac{q-1}{2}}}\limits_{q},\underbrace{1/\beta,\hdots,1/\beta}\limits_{q},\hdots,\underbrace{1/\beta,\hdots,1/\beta}\limits_{q}).$  Finally, the result follows from Lemma \ref{lem:char2}.
\end{proof}

\begin{cor} Let $N_1,K_1$ and $d_1$ be the length, dimension and minimum distance of the Hermitian $\ell$-$dim$ hull codes defined in Propositions \ref{thm:elliptic}-\ref{thm:hermitian-2} and Theorem \ref{thm:hermitian-3}. Then there exist Hermitian $\ell$-$dim$ hull codes with parameters $[N_1-s_1,K_1-s_1,\ge d_1-s_1]_{q^2}$ for $1\le s_1\le K_1-1$ and for $1\le \ell \le K_1-s_1-1.$
\end{cor}

\section{Application to constructions of EAQECCs}\label{section:application}
In this section, we construct quantum codes from Hermitian hulls of linear codes. 

By combining Lemma \ref{lem:Q-construction} and Lemma \ref{lem:hull-H} together, $q$-ary EAQECCs can be constructed from a classical $q^2$-linear code as follows.
\begin{lem}\textnormal{(\cite{GJG18})}\label{lem:GJG18-construction}
Let $C$ be a linear code with parameters $[n, k, d]_{q^2}$ and $C^{\perp_h}$ its Hermitian dual  with parameters $[n,k,d']_{q^2}$. Assume that $\dim (Hull_h(C))=\ell$. Then there exist an $[[n, k-\ell, d;n-k-\ell]]_{q}$ EAQECC and an $[[n, n-k-\ell, d';k-\ell]]_{q}$ EAQECC.
\end{lem}

It is well-known the (Hermitian) dual of an MDS linear code is again an MDS linear code. It is easy to check that if a code $C$ is a Hermitian $\ell$-$dim$ hull code, then so is its dual $C^{\perp_h}$. By applying the above lemma to the Hemitian $\ell$-$dim$ hull $[n,k,d]_{q^2}$ codes constructed in the previous section, we obtain the following result.

\begin{thm} \label{thm:Q-MDS}Let $q=p^m$ and $1\le k\le \lfloor \frac{n+q-1}{q+1} \rfloor$. Then there exist MDS EAQECCs with parameters $[[n-s,k-s-\ell,n-k+1;n-k-\ell]]_{q}$ and $[[n-s,n-k-\ell,k-s+1;k-s-\ell]]_{q}$ for $0\le s\le k-1$ and $0\le \ell \le k-s$ if one of the conditions holds
\begin{enumerate}
\item $n=q^2+1$, $k\le q$, $k\not=q-1$;
 $n=r(q-1)+1$, $k\le (q+r-1)/2$, $q\equiv r-1\mod{2r}$;$n=(q^2+2)/3$, $3|(q+1)$, $k\le (2q-1)/3$;
\item $n=tq$, $1\le t\le q$, $ k\le \lfloor\frac{tq+q-1}{q+1} \rfloor$;$n=t(q+1)+2$, $1\le t\le q-1$, $k\le t +1$, $(p,t,k)\not=(2,q-1,q-1)$;
\item $(n-1)|(q^2-1)$, $ k\le \lfloor\frac{n+q-1}{q+1} \rfloor$;
$n=(t+1)N+i$, $1\le i\le 2$, $N|(q^2-1)$, $n_2=\frac{N}{\gcd (N,q+1)}$, $1\le t\le \frac{q-1}{n_2}-1$, $ k\le \lfloor\frac{n+q-1}{q+1} \rfloor$.
\end{enumerate}
\end{thm}

%

\begin{rem}
 It should be noted that the families in Theorem \ref{thm:Q-MDS} 2) were obtained in \cite{FangFuLiZhu} by different approach
for $(n,s)=(tq,0)$ and $(n,s)=(t(q+1),0)$. 
However, their method does not allow to puncture coordinates to obtain new parameters of  EAQECCs as our method does.
\end{rem}

\begin{thm} \label{thm:Q-AMDS}Let $q=p^m$, $(n-1)|(q^2-1)$ and $1\le k\le \lfloor \frac{n+q-1}{q+1} \rfloor$.
\begin{enumerate}
\item If $(n-1)|(k+1)(q+1)$, then there exists an AMDS EAQECC with parameters $[[n+1,k+2-\ell,n-k-1;n-k-1-\ell]]_{q}$ for $0\le \ell \le k+2$
\item   If $(n-1)|k(q+1)$, then there exists an EAQECC with parameters $[[n+i-1,k+i-\ell,n-k-i+1;n-k-1-\ell]]_{q}$ for $0\le \ell \le k+i$. Moreover, there exists an AMDS EAQECC with parameters $[[q^2+1,k-\ell,q^2+1-k;q^2+1-k-\ell]]_{q}$ for $k=q+2,\hdots, 2q-2$, $0\le \ell \le k$.
\end{enumerate}
\end{thm}

\begin{thm}\label{thm:Q-family-new-mds} Let $q=p^m$, $(n-1)|(q^2-1)$ and $k=\lfloor \frac{n+q-1}{q+1}\rfloor.$  Then
\begin{enumerate}
\item if $(n-1)|(k+1)(q+1)$, then there exist MDS EAQECCs with parameters $[[n,k+2-\ell,n-k-1;n-k-2-\ell]]_{q}$ and $[[n,n-k-2,k+3;k+2-\ell]]_{q}$, where $\ell=(k+1)$;
\item if $(n-1)|k(q+1)$, then there exist MDS EAQECCs with parameters $[[n,k+i-\ell,n-k-i+1;n-k-i-\ell]]_{q}$ and $[[n,n-k-i-\ell,k+i+1;k+i-\ell]]_{q}$ for any $1\le i\le q$, where $\ell=k+i-\sharp I$;
 Moreover, 
\begin{enumerate}
\item there exist MDS EAQECCs with parameters $[[q^2,k+i-\ell,q^2-k-i+1;q^2-k-i-\ell]]_{q}$ and $[[q^2,q^2-k-i-\ell,k+i+1;k+i-\ell]]_{q}$ for $1\le i\le  q$, where $\ell= k+i-1$;
\item there exist MDS EAQECCs with parameters $[[2(q+1)+1,k+i-\ell,2(q+1)+1-k-i+1;2(q+1)+1-k-i-\ell]]_{q}$ and $[[2(q+1)+1,2(q+1)+1-k-i-\ell,k+i+1;k+i-\ell]]_{q}$ for $1\le i\le q$, where $\ell= (k+i-\lfloor \frac{i-1}{2} \rfloor -1)$.
\end{enumerate}
\end{enumerate}
\end{thm}

\begin{thm}  Let $q$ be an odd prime power, $N|(q^2-1)$, $n_2=\frac{N}{\gcd (N,q+1)}$, $1\le t\le \frac{q-1}{n_2}-1$. Put $n=(t+1)N+1$, $ k= \lfloor\frac{n+q-1}{q+1} \rfloor$, $k_i=k+i$ and $n_i=n+i$. Then 
\begin{enumerate}
\item there exist AMDS EAQECCs with parameters $[[n+2,k+2,\ge n-k;n-k-\ell]]_{q}$ and $[[n+2,n-k-\ell,\ge k+2;k+2-\ell]]_{q}$ for $0\le \ell \le k+2$;
\item  there exist EAQECCs with parameters $[[n+i,k+i-\ell,\ge n-k-i+1;n-k-\ell]]_{q}$ code for $0\le \ell \le k+i$, $1\le i\le q$;
\item there exist EAQECCs with parameters $[[n,2k-1-\ell,\ge n-k-q+2;n-2k+1-\ell]]_{q}$ $0\le \ell \le 2k-1$;
\item  there exist MDS EAQECCs with parameters $[[n,k+i-\ell,n-k-i+1;n-k-i-\ell]]_{q}$ and $[[n,n-k-i-\ell,k+i+1;k+i-\ell]]_{q}$ with $\ell=k+i-\sharp I$ for $1\le i\le q$. 
\end{enumerate}
\end{thm}

\begin{thm}\label{thm:Q-elliptic}
Let $q=2^s, m=2s$ and $2k+2\le n\le \lfloor q+\sqrt{2q}-5\rfloor$. Then there exist AMDS EAQECCs with parameters $[[n,k-\ell,\le n-k;n-k-\ell]]_{q}$ and $[[n,n-k-\ell,k;k-\ell]]_{q}$ for $0\le \ell \le k.$

\end{thm}

\begin{thm}\label{thm:Q-hyper-elliptic}
Let $q=2^m, m\ge 2$ and $q\le k\le \lfloor \frac{2q^2+2q-1}{q+1} \rfloor$. Then 
there exist EAQECCs with parameters $[[2q^2,k-q/2-\ell,\ge 2q^2-k+1;2q^2-k+q/2-\ell]]_{q}$ and $[[2q^2,2q^2+q/2-k,\ge k+1-q;2q^2,k-q/2-\ell]]_{q}$ for $0\le \ell \le k-q/2$.
\end{thm}

\begin{thm}\label{thm:Q-hermitian}
Let $q=p^m\ge 4$ and $q(q-1)\le k\le \lfloor \frac{n+q(q-1)-1}{q+1} \rfloor$. Then
there exist EAQECCs with parameters $[[Nq,k-\frac{q(q-1)}{2}-\ell,\ge Nq-k+1;Nq-k+\frac{q(q-1)}{2}-\ell ]]_{q}$ and $[[Nq,Nq+q(q-1)/2-k-\ell,\ge k+1-q(q-1);k-q(q-1)/2-\ell]]_{q}$ for $0\le \ell \le k-\frac{q(q-1)}{2}$ if one of the following conditions holds
\begin{enumerate}
\item $n=q^2+1$, $k\le q$, $k\not=q-1$;
 $n=r(q-1)+1$, $k\le (q+r-1)/2$, $q\equiv r-1\mod{2r}$;$n=(q^2+2)/3$, $3|(q+1)$, $k\le (2q-1)/3$;
\item $n=tq$, $1\le t\le q$, $ k\le \lfloor\frac{tq+q-1}{q+1} \rfloor$;$n=t(q+1)+2$, $1\le t\le q-1$, $k\le t +1$, $(p,t,k)\not=(2,q-1,q-1)$;
\item $(n-1)|(q^2-1)$;
$n=(t+1)N+i$, $1\le i\le 2$, $N|(q^2-1)$, $n_2=\frac{N}{\gcd (N,q+1)}$, $1\le t\le \frac{q-1}{n_2}-1$.
\end{enumerate}
\end{thm}
\begin{thm} Let $q=p^m$ be an odd prime power, $1<t< q+1$ and $t|(q+1)$. Put $g=\frac{(q-1)(q+1-t)}{2t}$ and $s={q}\left(\frac{\lfloor q^2-2\rfloor}{t}+2\right)$. Then for $2g \le k\le \lfloor\frac{s+q-1}{q+1}\rfloor$ and $0 \le \ell \le k-g$, there exist EAQECCs with parameters $[[s,k-g-\ell,\ge s-k+1;s-k+g-\ell]]_{q}$ and $[[s,s-k+g-\ell,\ge k-2g+1;k-g-\ell]]_{q}$.
\end{thm}

\begin{thm}\label{thm:Q-punctured} Let $N_1,K_1$ and $d_1$ be the length, dimension and minimum distance of the Hermitian $\ell$-$dim$ hull codes defined in Propositions \ref{thm:elliptic}-\ref{thm:hermitian-2} and Theorem \ref{thm:hermitian-3}. Then there exist EAQECCs with parameters $[[N_1-s_1,K_1-s_1-\ell,\ge d_1-s_1;N_1-K_1-\ell]]_{q}$ for $1\le s_1\le K_1-1$ and for $1\le \ell \le K_1-s_1.$
\end{thm}

\begin{table}
\centering
\caption{Some new MDS and AMDS EAQECCs from Theorem \ref{thm:Q-MDS} with $s=0,1$ and Theorem \ref{thm:Q-family-new-mds} 2) with $i=2$, $^*$: new parameters in \cite{FangFuLiZhu}
}
{\tiny
$$
\begin{array}{cccccc}
\ell &[n,k]_{q^2}&\text{MDS EAQECCs}&\text{MDS EAQECCs}&\text{Construction}\\
\hline
\hline
1&[15,4]_{8^2}&[[ 15, 10, 5;3 ]]_{8}^*&[[ 14, 10, 4;2 ]]_{8}^*&\text{Theorem \ref{thm:Q-MDS} 1)}\\

2&&[[ 15, 9, 5;2 ]]_{8}^*&[[ 14, 9, 4;1 ]]_{8}^*\\

1&[22,5]_{8^2}&[[ 22, 16, 6;4 ]]_{8}^*&[[ 21, 16, 5;3 ]]_{8}^*\\

2&&[[ 22, 15, 6;3 ]]_{8}^*&[[ 21, 15, 5;2 ]]_{8}^*\\

1&[29,5]_{8^2}&[[ 29, 23, 6;4 ]]_{8}^*&[[ 28, 23, 5;3 ]]_{8}^*\\

2&&[[ 29, 22, 6;3 ]]_{8}^*&[[ 28, 22, 5;2 ]]_{8}^*\\

1&[36,6]_{8^2}&[[ 36, 29, 7;5 ]]_{8}^*&[[ 35, 29, 6;4 ]]_{8}^*\\

2&&[[ 36, 28, 7;4 ]]_{8}^*&[[ 35, 28, 6;3]]_{8}^*\\

1&[43,6]_{8^2}&[[ 43, 36, 7;5 ]]_{8}^*&[[ 42, 36, 6;4 ]]_{8}^*\\

2&&[[ 43, 35, 7;4 ]]_{8}^*&[[ 42, 35, 6;3 ]]_{8}^*\\

1&[50,7]_{8^2}&[[ 50, 42, 8;6 ]]_{8}^*&[[ 49, 42, 7;5 ]]_{8}\\

2&&[[ 50, 41, 8;5 ]]_{8}^*&[[ 49, 41, 7;4 ]]_{8}\\

1&[57,7]_{8^2}&[[ 57, 49, 8;6 ]]_{8}^*&[[ 56, 49, 7;5 ]]_{8}\\

2&&[[ 57, 48, 8;5 ]]_{8}^*&[[ 56, 48, 7;4 ]]_{8}\\

1&[72,7]_{9^2}& [[72,64,8;6]]_{9}& [[71,64,7;5]]_{9}^*&\text{Theorem \ref{thm:Q-MDS} 2)}\\

2&& [[72,63,8;5]]_{9}& [[71,63,7;4]]_{9}^*\\

3&& [[72,62,8;4]]_{9}& [[71,62,7;3]]_{9}^*\\

4& &[[72,61,8;3]]_{9}&[[71,61,7;2]]_{9}^*\\

5&& [[72,60,8;2]]_{9}& [[71,60,7;1]]_{9}^*\\

6& &[[72,59,8;1]]_{9} &[[71,59,7;0]]_{9}^*\\
  
1&[96,8]_{11^2}& [[96,87,9;7]]_{11} & [[95,87,8;6]]_{11}^*& \text{Theorem \ref{thm:Q-MDS} 3)}\\

2&& [[96,86,9;6]]_{11} & [[95,86,8;5]]_{11}^*\\

3&& [[96,85,9;5]]_{11} & [[95,85,8;4]]_{11}^*\\

4&& [[96,84,9;4]]_{11} & [[95,84,8;3]]_{11}^*\\

5&& [[96,83,9;3]]_{11} & [[95,83,8;2]]_{11}^*\\

6& &[[96,82,9;2]]_{11} &[[95,82,8;1]]_{11}^*\\

7&& [[96,81,9;1]]_{11} & [[95,81,8;0]]_{11}^*\\

1&[ 16, 3 ]_{4^2}&[[ 16, 7, 6; 1 ]]_{4}^*&&\text{Theorem \ref{thm:Q-family-new-mds} 2)} \\

1&[ 25, 4 ]_{5^2}&[[ 25, 14, 7; 1 ]]_{5}^*&\\

1&[ 49, 6 ]_{7^2}&[[ 49, 34, 9; 1 ]]_{7}^*&\\

1&[ 64, 7 ]_{8^2}&[[ 64, 47, 10; 1 ]]_{8}^*&\\

1&[ 81, 8 ]_{9^2}&[[ 81, 62, 11; 1 ]]_{9}^*&\\

1&[ 121, 10 ]_{11^2}&][ 121, 98, 13; 1 ]]_{11}^*&\\

\end{array}
$$
}
\label{table:mds1}
\end{table}

\begin{table}
\centering
\caption{Some new MDS and MDS EAQECCs from Theorem \ref{thm:Q-MDS} 3) with $s=0,1$ and Theorem \ref{thm:Q-AMDS} with $i=2$, $^*$: new parameters in \cite{FangFuLiZhu}
}
{\tiny
$$
\begin{array}{cccccc}
\ell &[n,k]_{q^2}&\text{MDS EAQECCs}&\text{MDS EAQECCs}&\text{AMDS EAQECCs}\\
\hline
\hline

1&[16,3]_{4^2}&[[ 16, 12, 4 ; 2 ]]_4&[[ 15, 12, 3; 1 ]]_4^*&[[ 17, 11, 5 ; 4 ]]_4^*\\

2&&[[ 16, 11, 4 ; 1 ]]_4&[[ 15, 11, 3 ; 0 ]]_4^*&[[ 17, 10, 5 ; 3 ]]_4^*\\

1&[25,4]_{5^2}&[[ 25, 20, 5 ; 3 ]]_5&[[ 24, 20, 4 ; 2 ]]_5^*&[[ 26, 19, 6 ; 5 ]]_5^*\\

2&&[[ 25, 19, 5 ; 2 ]]_5&[[ 24, 19, 4 ; 1 ]]_5^*&[[ 26, 18, 6 ; 4 ]]_5^*\\

3&&[[ 25, 18, 5 ; 1 ]]_5&[[ 24, 18, 4 ; 0 ]]_5^*&[[ 26, 17, 6 ; 3 ]]_5^*\\

1&[49,6]_{7^2}&[[ 49, 42, 7 ; 5 ]]_7&[[ 48, 42, 6 ; 4]]_7^*&[[ 50, 42, \ge 8 ; 6 ]]_7^*\\

2&&[[ 49, 41, 7 ; 4 ]]_7&[[ 48, 41, 6 ; 3 ]]_7^*&[[ 50, 40, \ge8 ; 6 ]]_7^*\\

3&&[[ 49, 40, 7 ; 3 ]]_7&[[ 48, 40, 6; 2 ]]_7^*&[[ 50, 39, \ge8 ; 5 ]]_7^*\\

4&&[[ 49, 39, 7 ; 2 ]]_7&[[ 48, 39, 6 ; 1 ]]_7^*&[[ 50, 38, \ge8 ; 4 ]]_7^*\\

5&&[[ 49, 38, 7 ; 1 ]]_7&[[ 48, 38, 6 ; 0 ]]_7^*&[[ 50, 37, \ge8 ; 3 ]]_7^*\\

1&[64,7]_{8^2}&[[ 64, 56, 8 ; 6 ]]_8&[[ 63, 56, 7 ; 5 ]]_8^*&[[ 65, 55, \ge9 ; 8 ]]_8^*\\

2&&[[ 64, 55, 8 ; 5 ]]_8&[[ 63, 55, 87; 4 ]]_8^*&[[ 65, 54, \ge9 ; 7 ]]_8^*\\

3&&[[ 64, 54, 8 ; 4 ]]_8&[[ 63, 54, 7 ; 3 ]]_8^*&[[ 65, 53, \ge9 ; 6 ]]_8^*\\

4&&[[ 64, 53, 8 ; 3 ]]_8&[[ 63, 53, 7 ; 2 ]]_8^*&[[ 65, 52, \ge9 ; 5 ]]_8^*\\

5&&[[ 64, 52, 8 ; 2 ]]_8&[[ 63, 52, 7 ; 1 ]]_8^*&[[ 65, 51, \ge9 ; 4 ]]_8^*\\

6&&[[ 64, 51, 8 ; 1 ]]_8&[[ 63, 51, 7 ; 0 ]]_8^*&[[ 65, 50, \ge9 ; 3 ]]_8^*\\   

1&[81,8]_{9^2}&[[ 81, 72, 9 ; 7 ]]_{9}&[[ 80, 72, 8 ; 6 ]]_{9}^*&[[ 82, 71, \ge10 ; 9 ]]_{9}^*\\

2&&[[ 81, 71, 9 ; 6 ]]_{9}&[[ 80, 71, 8 ; 5 ]]_{9}^*&[[ 82, 70, \ge10 ; 8 ]]_{9}^*\\

3&&[[ 81, 70, 9 ; 5 ]]_{9}&[[ 80, 70, 8 ; 4 ]]_{9}^*&[[ 82, 69, \ge10 ; 7 ]]_{9}^*\\

4&&[[ 81, 69, 9 ; 4 ]]_{9}&[[ 80, 69, 8; 3 ]]_{9}^*&[[ 82, 68, \ge10 ; 6 ]]_{9}^*\\

5&&[[ 81, 68, 9 ; 3 ]]_{9}&[[ 80, 68, 8 ; 2 ]]_{9}^*&[[ 82, 67, \ge10 ; 5 ]]_{9}^*\\

6&&[[ 81, 67, 9 ; 2 ]]_{9}&[[ 80, 67, 8 ; 1 ]]_{9}^*&[[ 82, 66, \ge10 ; 4 ]]_{9}^*\\

7&&[[ 81, 66, 9 ; 1 ]]_{9}&[[ 80, 66, 8 ; 0 ]]_{9}^*&[[ 82, 65, \ge10 ; 3 ]]_{9}^*\\

1&[121,10]_{11^2}&[[ 121, 110, 11 ; 9 ]]_{11}&[[ 120, 110, 10 ; 8 ]]_{11}^*&[[ 122, 109, \ge12 ; 11 ]]_{11}^*\\

2&&[[ 121, 109, 11 ; 8 ]]_{11}&[[ 120, 109, 10 ; 7 ]]_{11}^*&[[ 122, 108, \ge12 ; 10 ]]_{11}^*\\

3&&[[ 121, 108, 11 ; 7 ]]_{11}&[[ 120, 108, 10 ; 6 ]]_{11}^*&[[ 122, 107, \ge12 ; 9 ]]_{11}^*\\

4&&[[ 121, 107, 11 ; 6 ]]_{11}&[[ 120, 107, 10 ; 5 ]]_{11}^*&[[ 122, 106, \ge12 ; 8 ]]_{11}^*\\

5&&[[ 121, 106, 11 ; 5 ]]_{11}&[[ 120, 106, 10 ; 4 ]]_{11}^*&[[ 122, 105, \ge12 ; 7]]_{11}^*\\

6&&[[ 121, 105, 11 ; 4 ]]_{11}&[[ 120, 105, 10 ; 3 ]]_{11}^*&[[ 122, 104, \ge12 ; 6 ]]_{11}^*\\

7&&[[ 121, 104, 11 ; 3 ]]_{11}&[[ 120, 104, 10 ; 2 ]]_{11}^*&[[ 122, 103, \ge12 ; 5 ]]_{11}^*\\

8&&[[ 121, 103, 11 ; 2 ]]_{11}&[[ 120, 103, 10 ; 1 ]]_{11}^*&[[ 122, 102, 1\ge2 ; 4 ]]_{11}^*\\

9&&[[ 121, 102, 11 ; 1 ]]_{11}&[[ 120, 102, 10 ; 0 ]]_{11}^*&[[ 122, 101, \ge12 ; 3 ]]_{11}^*\\

\end{array}
$$
}
\label{table:mds2}
\end{table}

\begin{table}
\centering
\caption{Some new EAQECCs from Theorem \ref{thm:Q-elliptic} and Theorem \ref{thm:Q-hyper-elliptic}, $^*$: new parameters
}
{\tiny
$$
\begin{array}{cccccc}
\ell &[n,k]_{q^2}&\text{EAQECCs}&\text{Punctured EAQECCs}&\text{Construction}\\
\hline
\hline

1&[32,6,25]_{4^2}& [[32,5,25;25]]_4^*& [[31,4,\ge 24;25]]_4^*&\text{Theorems \ref{thm:Q-hyper-elliptic}, \ref{thm:Q-punctured}}\\

2&&[[32,4,25;24]]_4^*&[[31,3,\ge 24;24]]_4^*&\\

3&&[[32,3,25;23]]_4^*&[[31,2,\ge24;23]]_4^*&\\

4&&[[32,2,25;22]]_4^*&[[31,1,\ge24;22]]_4^*&\\

5&&[[32,1,25;21]]_4^*&[[31,0,\ge24;21]]_4^*&\\

1&[80, 9, 71]_{8^2}&[[80, 8, 71;70]]_8^*&[[79, 7, \ge70;70]]_8^*&\text{Theorems \ref{thm:Q-elliptic},
\ref{thm:Q-punctured}}\\

2&&[[80, 7, 71;69]]_8^*&[[79, 6, \ge70;69]]_8^*&\\

3&&[[80, 6, 71;68]]_8^*&[[79, 5, \ge70;68]]_8^*&\\

4&&[[80, 5, 71;67]]_8^*&[[79, 4, \ge70;67]]_8^*&\\

5&&[[80, 4, 71;66]]_8^*&[[79, 3, \ge70;66]]_8^*&\\

6&&[[80, 3, 71;65]]_8^*&[[79, 2, \ge70;65]]_8^*&\\

7&&[[80, 2, 71;64]]_8^*&[[79, 1, \ge70;64]]_8^*&\\

8&&[[80, 1, 71;63]]_8^*&[[79, 0, \ge70;63]]_8^*&\\

1&[128, 12, 113]_{8^2}&[[128, 11, 113;115]]_8^*&[[127, 10, \ge112;115]]_8^*&\text{Theorem \ref{thm:Q-hyper-elliptic}, \ref{thm:Q-punctured}}\\

2&&[[128, 10, 113;114]]_8^*&[[127, 9, \ge112;114]]_8^*&\\

3&&[[128, 9, 113;113]]_8^*&[[127, 8, \ge112;113]]_8^*&\\

4&&[[128, 8, 113;112]]_8^*&[[127, 7, \ge111;112]]_8^*&\\

5&&[[128, 7, 113;111]]_8^*&[[127, 6, \ge112;111]]_8^*&\\

6&&[[128, 6, 113;110]]_8^*&[[127, 5, \ge112;110]]_8^*&\\

7&&[[128, 5, 113;109]]_8^*&[[127, 4, \ge112;109]]_8^*&\\

8&&[[128, 4, 113;108]]_8^*&[[127, 3, \ge112;108]]_8^*&\\

9&&[[128, 3, 113;107]]_8^*&[[127, 2, \ge112;107]]_8^*&\\

10&&[[128, 2, 113;106]]_8^*&[[127, 1, \ge112;106]]_8^*&\\

11&&[[128, 1, 113;105]]_8^*&[[127, 0, \ge112;105]]_8^*&\\

\end{array}
$$
}
\label{table:non-mds}
\end{table}


\end{document}